%% file: source/260416_arxiv_v2.tex
\documentclass[aps,prx,twocolumn,reprint]{revtex4-2}

\usepackage{graphicx}
\usepackage{dcolumn}
\usepackage{bm}
\usepackage{amsmath,amssymb,mathtools,braket}
\usepackage{color,comment,footnote}

\usepackage{algorithmic}
\usepackage{algorithm}

\usepackage{tcolorbox}
\usepackage{amsthm}
\theoremstyle{plain}
\newtheorem{thm}{Theorem}
\newtheorem{cor}{Corollary}
\newtheorem{lem}{Lemma}

\theoremstyle{definition}
\newtheorem{dfn}{Definition}

\usepackage[whole]{bxcjkjatype}

\usepackage{hyperref}
\usepackage{xcolor}
\hypersetup{
colorlinks=true,
citecolor=blue,
linkcolor=blue,
urlcolor=blue,
}

\newcommand{\mc}{\mathcal}
\newcommand{\mO}{\mathcal{O}}

\newcommand{\tr}{\text{Tr}}

\newcommand{\mQ}{\mathcal{Q}}

\newcommand{\mH}{\mathcal{H}}
\newcommand{\mP}{\mathcal{P}}

\newcommand{\bma}{\bm{\alpha}}
\newcommand{\bmb}{\bm{\beta}}
\newcommand{\bmt}{\bm{\theta}}
\newcommand{\bme}{\bm{\epsilon}}
\newcommand{\bmn}{\bm{n}}

\newcommand{\Qb}{Q_\text{bit}}

\begin{document}

\title{Resource-efficient equivariant quantum convolutional neural networks}

\author{Koki Chinzei}\thanks{chinzei.koki@fujitsu.com}
\author{Quoc Hoan Tran}
\author{Yasuhiro Endo}
\author{Hirotaka Oshima}
\affiliation{Quantum Laboratory, Fujitsu Research, Fujitsu Limited,
	4-1-1 Kawasaki, Kanagawa 211-8588, Japan}
\date{\today}

\begin{abstract}
Equivariant quantum neural networks (QNNs) are promising variational models that exploit symmetries to improve machine learning capabilities.
Despite theoretical developments in equivariant QNNs, their implementation on near-term quantum devices remains challenging due to limited computational resources.
This study proposes a resource-efficient model of equivariant quantum convolutional neural networks (QCNNs) called equivariant split-parallelizing QCNN (sp-QCNN). 
Using a group-theoretical approach, we encode general symmetries into our model beyond the translational symmetry addressed by previous sp-QCNNs. 
We achieve this by splitting the circuit at the pooling layer while preserving symmetry.
This splitting structure effectively parallelizes QCNNs to improve measurement efficiency in estimating the expectation value of an observable and its gradient by order of the number of qubits. 
Our model also exhibits high trainability and generalization performance, including the absence of barren plateaus. 
Numerical experiments demonstrate that the equivariant sp-QCNN can be trained and generalized with fewer measurement resources than a conventional equivariant QCNN in a noisy quantum data classification task.
Our results contribute to the advancement of practical quantum machine learning algorithms.
\end{abstract}

\maketitle

\section{Introduction}

Demonstrating practical quantum advantages is a significant challenge in quantum information science.
While several quantum algorithms, such as Shor's factoring~\cite{Shor1997-bt} and Grover's search~\cite{Grover1997-zk}, are believed to achieve this goal, they require a large-scale fault-tolerant quantum computer, which will prevent the early realization of quantum advantages~\cite{Gidney2021-ds}. 
Given this situation, numerous studies have focused on developing quantum algorithms that can be implemented on near-term quantum devices~\cite{Preskill2018-xr}. 
One promising approach is using variational quantum algorithms (VQAs), where the parameterized quantum circuit (PQC) is optimized to solve a given problem with quantum and classical computers~\cite{Cerezo2021-un}.
The VQAs apply to various tasks, such as finding the ground and excited states of Hamiltonians~\cite{Peruzzo2014-oj}, learning unknown patterns and relationships from data (i.e., quantum machine learning, QML)~\cite{Farhi2018-nt, Mitarai2018-ap, Benedetti2019-ph, Schuld2020-vz}, and solving combinatorial optimization problems~\cite{Farhi2014-as}.

Despite fundamental interests in VQAs, several obstacles exist that prevent them from achieving quantum advantages.
The most critical issue in VQAs is the poor trainability of PQCs.
This is often due to barren plateaus, where the exponentially flat landscape of the cost function causes the optimization process to fail, preventing quantum speedups~\cite{McClean2018-qf, Cerezo2021-tq, Ortiz-Marrero2021-lx, Holmes2022-uk}.
While previous studies have proposed many PQCs that can avoid the barren plateau phenomenon by utilizing polynomial-size subspaces instead of the entire Hilbert space~\cite{Larocca2024-vh}, such PQCs face another critical challenge: classical simulability~\cite{Cerezo2023-hz, Bermejo2024-dg}.
Recent studies indicate that provably barren plateau-free models can be efficiently simulated on classical computers when applied to locally easy datasets, highlighting the fundamental challenge in achieving exponential quantum speedups.
Nonetheless, some potential quantum advantages may still exist. 
Since classical simulability has been demonstrated only for locally easy datasets, investigating QNN performance on locally nontrivial datasets presents an intriguing opportunity to uncover potential quantum benefits~\cite{Bermejo2024-dg}.
In addition, classical simulability does not preclude the possibility of polynomial improvements in training and inference efficiency~\cite{Cerezo2023-hz}.
Further exploring and understanding the potentials and limitations of VQAs is important for developing more effective quantum algorithms.

In these situations surrounding VQAs, symmetry plays a vital role in maximizing their potential.
Geometric quantum machine learning (GQML) is a methodology that leverages the geometric structure (e.g., symmetry) of data to solve machine learning tasks with high scalability~\cite{Bronstein2021-is, Verdon2019-pi, Zheng2023-ze, Larocca2022-mj, Meyer2023-vx, Skolik2023-lq, Ragone2022-va, Nguyen2024-nt, Sauvage2024-vd, Anschuetz2022-gb, Anschuetz2024-ho}.
In GQML, symmetry is encoded into an equivariant quantum neural network (QNN) as an inductive bias, reducing the parameter space to search.
This leads to high trainability and generalization (although it does not necessarily resolve the classical simulability issue~\cite{Cerezo2023-hz}).
For instance, permutation equivariant QNNs were theoretically proven not to show barren plateaus and to improve the generalization performance~\cite{Schatzki2024-il}.
The theory of equivariant QNNs has been extended to several circuit models, including the quantum convolutional neural network (QCNN)~\cite{Grant2018-re, Cong2019-ov, Nguyen2024-nt}, which is a representative hierarchical QNN with high trainability and feasibility~\cite{Pesah2021-py}.

\begin{figure*}[t]
	\centering
	\includegraphics[width=\linewidth]{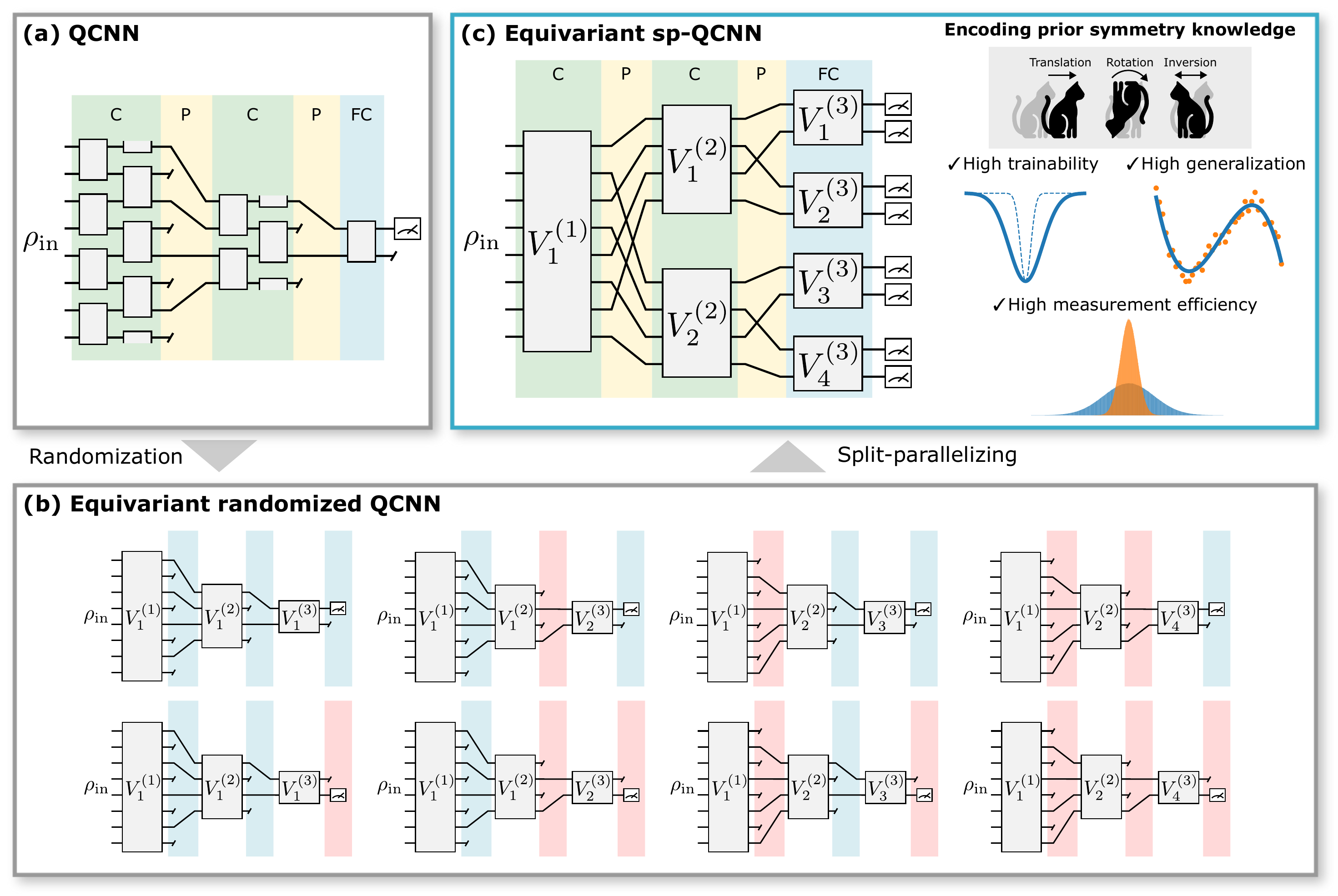}
	\caption{
Basic structures of (a) conventional QCNN~\cite{Cong2019-ov}, (b) equivariant randomized QCNN~\cite{Nguyen2024-nt}, and (c) equivariant sp-QCNN.
(a) In the conventional QCNN, some qubits are discarded at each pooling layer, and the remaining qubits are measured at the end of the circuit.
(b) The equivariant randomized QCNN randomly chooses which qubits to discard at each pooling layer to ensure the equivariance.
The blue (red) bars indicate the pooling layer that discards even-(odd-)numbered qubits. 
(c) The equivariant sp-QCNN splits the circuit to impose the equivariance rather than discarding the qubits.
This splitting structure with equivariance results in high trainability, generalization, and measurement efficiency.
	}
	\label{fig: illust}
\end{figure*}

In practice, the high computational cost of QNNs, i.e., many repetitions of measurements during training, makes it challenging to implement them on resource-limited near-term quantum devices~\cite{Schuld2022-wb, Liu2023-pv, Chinzei2025-pa}.
Solving large-scale problems requires large amounts of training data for high generalization, expressive quantum models for high accuracy, and many measurement shots to accurately estimate expectation values, resulting in substantial computational costs.
Recently, an alternative model of equivariant QCNNs, split-parallelizing QCNN (sp-QCNN), has been proposed to reduce the measurement cost in expectation value estimation~\cite{Chinzei2024-nm}.
The sp-QCNN splits the quantum circuit into multiple branches at the pooling layers instead of discarding the qubits (Fig.~\ref{fig: illust}).
This splitting structure maximizes the use of the qubit resource and is thus more efficient than the conventional QCNNs in terms of the number of measurement shots required to achieve a certain accuracy in estimating the expectation value of an observable.
Therefore, the sp-QCNN is especially promising for near-term quantum devices with limited computational resources.

Nevertheless, the previous study~\cite{Chinzei2024-nm} has focused only on situations where the input data is translationally symmetric, severely limiting its applicability. 
For instance, chemical molecules, a major target in quantum computation, do not exhibit translational symmetry; instead, they possess other complex symmetries such as rotations and inversions. 
Extending the sp-QCNN framework to these more general symmetry groups presents a nontrivial challenge. 
Due to its unique splitting-parallelizing structure, directly applying existing techniques for imposing equivariance in common QNN architectures~\cite{Meyer2023-vx, Nguyen2024-nt} is not feasible for sp-QCNNs. 
This inherent architectural difference necessitates the development of an entirely new theoretical and methodological framework to construct general equivariant sp-QCNNs. 
Overcoming this challenge is crucial for expanding the practical utility of sp-QCNNs to a wider array of scientifically relevant problems.

In this work, we propose the framework of equivariant sp-QCNNs for general symmetries beyond translational symmetry and theoretically elucidate their practical advantages.
Our model encodes general symmetries into the quantum circuit to ensure equivariance by splitting the circuit based on the symmetry.
This splitting structure effectively parallelizes QCNNs to improve the measurement efficiency by a factor of $\mO(n)$ for general symmetries compared to the conventional model, where $n$ is the number of qubits.
Furthermore, the equivariance leads to high trainability and generalization, as in the conventional equivariant QCNN.
In particular, our model inherits the high trainability of QCNNs and does not suffer from barren plateaus.
To leverage these advantages, we present a group-theoretical method for constructing equivariant sp-QCNNs, which provides the theoretical basis for general model design.
To demonstrate the high performance of equivariant sp-QCNNs, we apply our model to the classification task of noisy Heisenberg model ground states with square lattice symmetry.
The result shows that the high measurement efficiency of the equivariant sp-QCNN suppresses the statistical error in estimating the expectation value, accelerating the training process compared to a conventional equivariant QCNN. 
Moreover, the equivariant sp-QCNN can achieve high classification accuracy with fewer training data than a non-equivariant model, indicating the high generalization of our model.

The remainder of this paper is organized as follows.
First, Sec.~\ref{sec: preliminaries} provides a concise overview of QCNNs and equivariant QNNs.
Section~\ref{sec: sp-QCNN} formulates the equivariant sp-QCNN and shows its high measurement efficiency in estimating the expectation value of an observable and its gradient compared to the conventional equivariant QCNN.
We also prove the absence of barren plateaus in the sp-QCNN under modest assumptions that its subcircuits (equivalent to conventional QCNNs) do not exhibit barren plateaus.
Section~\ref{sec: construction} introduces a group-theoretical method for constructing the equivariant sp-QCNNs and provides the theoretical basis for general model design.
Section~\ref{sec: demo} numerically demonstrates the effectiveness of the equivariant sp-QCNN, showing that it outperforms conventional equivariant and non-equivariant QCNNs when the amounts of measurement resources and data are limited.
Finally, Sec.~\ref{sec: conclusions} summarizes our proposal and results and discusses potential future research directions.


\section{Quantum convolutional neural networks and equivariance} \label{sec: preliminaries}

This section concisely reviews QCNNs and equivariant QNNs.
In particular, we introduce the randomized model of equivariant QCNNs.

\subsection{Quantum convolutional neural network}

The QCNN is a promising QNN architecture inspired by the classical convolutional neural network (CNN)~\cite{LeCun2015-yl, Krizhevsky2012-ya, LeCun1995-yw}. 
This model applies to a variety of machine learning tasks, such as quantum phase recognition~\cite{Herrmann2022-bs, Liu2022-uu, Monaco2023-xr}, high-energy event classification~\cite{Chen2022-gy, Nagano2023-ww}, and quantum error-correcting code optimization~\cite{Cong2019-ov}.
Similar to the CNN, the QCNN has a hierarchical structure consisting of three types of layers: convolutional, pooling, and fully connected layers~\cite{Grant2018-re, Cong2019-ov}.
The convolutional layers utilize local unitary gates to extract the local features from the input data,
while the pooling layers discard some qubits to coarse-grain the quantum information. 
In the originally proposed QCNN, the pooling layers are implemented by measuring some qubits and applying unitary gates, conditioned on the measurement outcomes, to the neighboring qubits. 
However, in this work, we implement the pooling layer solely with unitary gates, leveraging the equivalence between conditioned-gate operations and controlled-unitary operations~\cite{Nielsen2010-pf}.
After alternately applying the convolutional and pooling layers, one performs the fully connected transformation and then measures the remaining qubits to obtain an output. 
In supervised learning, the unitary gates in the circuit are optimized to represent the correct input-output relationship.

High feasibility and trainability are notable properties of QCNNs.
Since the number of qubits $n$ in the QCNN decreases exponentially with each pooling layer, the circuit depth is $\mO(\log n)$. 
This logarithmic depth leads to easy implementation in near-term quantum devices where the number of possible gate operations is limited.
Moreover, it is known that the QCNN does not suffer from barren plateaus~\cite{McClean2018-qf, Cerezo2021-tq, Ortiz-Marrero2021-lx, Holmes2022-uk}, the exponentially flat landscape of the cost function, due to the logarithmic circuit depth and the localities of unitary operations and observables~\cite{Pesah2021-py}. 
These properties imply the high scalability of QCNNs.

However, the high measurement resource requirement for training QCNNs presents practical difficulties.
Given that the QCNN has $\mO(n)+\mO(n/2)+\mO(n/4)+\cdots=\mO(n)$ parameters, it generally requires $\mO(nN_\text{train}N_\text{epoch}N_\text{shot})$ measurement shots for training in total, where $N_\text{train}$ is the number of training data, $N_\text{epoch}$ is the maximum epoch of training, and $N_\text{shot}$ is the number of measurement shots used per circuit.
Therefore, addressing large-scale problems that involve many qubits and a substantial amount of data is difficult in practice.
In this work, we demonstrate that the equivariant sp-QCNN can ideally reduce $N_\text{shot}$ by a factor of $\mO(1/n)$ for general symmetries, thus alleviating the measurement resource requirement.

\subsection{Equivariant QNN}

The GQML based on equivariant QNNs has recently emerged as a potential solution to some critical QML problems related to trainability and generalization~\cite{Bronstein2021-is, Verdon2019-pi, Zheng2023-ze, Larocca2022-mj, Meyer2023-vx, Skolik2023-lq, Ragone2022-va, Nguyen2024-nt, Sauvage2024-vd, Anschuetz2022-gb, Anschuetz2024-ho}.
It exploits the symmetry of a problem as an inductive bias and provides a problem-tailored circuit model, typically an equivariant QNN. 
For example, in supervised learning, the GQML assumes the label symmetry of the learning problem.
The label symmetry arises in various problems, such as image recognition, graph classification, and quantum physics.
To define the label symmetry, let $\rho$ be an input density matrix and $f(\rho)$ be a target function to be learned.
Then, the label symmetry $G$ is defined as
\begin{align}
    f(\rho) = f(U_g \rho U_g^\dag) \quad \forall g\in G, \,\forall \rho, \label{eq: label symmetry}
\end{align}
where $U_g$ is a unitary representation of the group $G$.

The GQML encodes this label symmetry into a quantum circuit as an equivariant QNN to improve trainability and generalization.
The equivariant QNN is a parametrized quantum circuit $U(\bmt)$ that is invariant under the action of $G$:
\begin{align}
    [U(\bmt),U_g]=0 \quad \forall g\in G.
\end{align}
This symmetry leads to the equivariance between input and output quantum states, $U(\bmt)[U_g \rho U_g^\dag]U^\dag(\bmt) = U_g[U(\bmt)\rho U^\dag(\bmt)]U_g^\dag$.
That is, applying the symmetry operation to the input state is identical to applying it to the output state.
Given this equivariant QNN and a $G$-symmetric observable $O$, the label symmetry of Eq.~\eqref{eq: label symmetry} always holds with $f(\rho)=\tr[U(\bmt)\rho U^\dag(\bmt) O]$, thus improving trainability and generalization.

The equivariant QNNs can be designed using several methods applicable to various QNN models, including the twirling method, the null space method, and the Choi operator method~\cite{Meyer2023-vx, Nguyen2024-nt}.
In QCNNs, however, the pooling layers generally break spatial symmetries (e.g., translation, rotation, and inversion of qubit positions) by discarding some qubits in the middle of the circuit, which prevents the straightforward implementation of spatially equivariant QCNNs.
To circumvent such difficulties, a randomized technique was introduced to impose the spatial equivariance of the pooling layer~\cite{Nguyen2024-nt}.
This technique randomly chooses which qubits to discard in each pooling layer for each measurement shot based on a given symmetry.
For example, for translational symmetry $T$ (e.g., $T\ket{10\cdots0}=\ket{010\cdots0}$), the equivariance is achieved by randomly choosing which to discard even- or odd-numbered qubits every measurement shot.
Then, the quantum operation for an input $\rho$ is given by $\rho \to (U_\text{e} \rho U_\text{e}^\dag + U_\text{o} \rho U_\text{o}^\dag)/2$, where $U_\text{e(o)}$ is the unitary operation acting on the even-(odd-)numbered qubits after the pooling layer with $TU_\text{e} T^\dag = U_\text{o}$ [see Fig.~\ref{fig: illust} (b)].
In other words, the randomized technique classically mixes the multiple QCNNs to ensure the equivariance.

The sp-QCNN is an alternative approach that coherently executes the multiple QCNNs to achieve the equivariance~\cite{Chinzei2024-nm}.
As shown in Fig.~\ref{fig: illust} (c), the sp-QCNN splits the circuit in the pooling layers instead of randomly selecting which qubits to discard.
This splitting structure maximally makes use of the qubit resource, improving the measurement efficiency in estimating the expectation value of an observable compared to the randomized method.
In Ref.~\cite{Chinzei2024-nm}, the high measurement efficiency of sp-QCNN for translationally symmetric data was demonstrated based on the effective parallelization of conventional QCNNs.
However, the previous work focused only on the translationally symmetric data, limiting the applicability of sp-QCNNs.
This work extends the concept of sp-QCNNs to equivariant sp-QCNNs and develops their theoretical framework for general symmetries.

\section{Equivariant sp-QCNN} \label{sec: sp-QCNN}

This section formulates the equivariant sp-QCNN and shows that it improves the measurement efficiency by a factor of $\mO(n)$ compared to the conventional randomized QCNN.
We also prove that the sp-QCNN does not suffer from barren plateaus under modest assumptions, as do conventional QCNNs.

\subsection{Model}

The equivariant sp-QCNN is a resource-efficient model of equivariant QCNNs, consisting of circuit splitting and unitary operations [see Fig.~\ref{fig: illust} (c)].
Here, let $\Qb=[n]$ be the set of qubits whose element corresponds to each qubit (throughout this paper, we denote $[a]=\{1,2,\cdots,a\}$ with an integer $a$).
We write the circuit splitting in the $\ell$th layer as a partition of $\Qb$:
\begin{align}
    \Qb=\bigsqcup_{i=1}^{s_\ell} Q_i^{(\ell)}, \label{eq: Qdecomp}
\end{align}
where $Q_i^{(\ell)}$ is a subset of the qubits representing the $i$th branch of the $\ell$th layer, $s_\ell$ is the number of branches in the $\ell$th layer, and $\sqcup$ denotes the disjoint union.
Note that the branches are not overlapped [i.e., $Q_i^{(\ell)}\cap Q_j^{(\ell)}=\varnothing$ for $i\neq j$] and that they cover all qubits [i.e., $\bigsqcup_i Q_i^{(\ell)}= \Qb$].
We assume that each branch in the $(\ell+1)$th layer is connected to a corresponding single branch in the $\ell$th layer:
\begin{align}
    & \text{$\forall i, \exists j$ s.t. $Q^{(\ell+1)}_i \subseteq Q^{(\ell)}_j$.}
\label{eq:cc} 
\end{align}
To clarify, each branch can be split but cannot be merged to another branch. 
As discussed later, this splitting structure contributes to the effective parallelization of QCNNs.

The unitary of the entire circuit is given by
\begin{align}
    U = \prod_{\ell=1}^{L} V^{(\ell)},
\end{align}
where $V^{(\ell)}$ is the unitary of the $\ell$th convolutional layer.
The convolutional layer $V^{(\ell)}$ is decomposed into $s_\ell$ unitaries $V_i^{(\ell)}$, each of which acts on $Q_i^{(\ell)}$, as
\begin{align}
    V^{(\ell)}=\prod_{i=1}^{s_\ell} V_i^{(\ell)}. \label{eq: Velldecomp}
\end{align}
Since each $V_i^{(\ell)}$ acts on different qubits, they are mutually commutative: $[V_i^{(\ell)}, V_j^{(\ell)}]=0$.

When the label symmetry $G$ of a problem is known in advance, the equivariance for the symmetry can help improve the trainability and generalization of the quantum model.
Here, the $G$-equivariance of sp-QCNN at each layer is defined as
\begin{align}
    [V^{(\ell)},U_g]=0 \quad \forall g\in G. \label{eq:Vell_sym}
\end{align}
Along with a $G$-symmetric observable $O$ (i.e., $[U_g,O]=0$ $\forall g\in G$), this equivariance leads to the label symmetry of $f(\rho) = \tr[U \rho U^\dag O]$ as $f(U_g \rho U_g^\dag)=f(\rho)$ for $\forall \rho$ and $\forall g\in G$.

This equivariance effectively reduces the expressivity without sacrificing accuracy, thereby improving trainability and generalization.
According to Ref.~\cite{Schatzki2024-il}, the improved trainability can be theoretically proven for the overparameterized circuit, where the variance of the cost function scales with the inverse of the expressivity~\cite{Ragone2024-hl, Fontana2024-ky}.
The equivariance reduces the expressivity, thus mitigating the decline in the variance of the cost function (e.g., barren plateaus)~\cite{Schatzki2024-il}.
In addition, some previous studies suggest that the equivariance may help to eliminate bad local minima~\cite{Wiersema2020-yv, Kim2021-ej, Kim2022-fm, Anschuetz2022-gb, Larocca2023-ll, Anschuetz2024-ho}.
On the other hand, the improved generalization can be shown in an information-theoretical way, where generalization errors are upper bounded by the logarithm of the $\epsilon$-covering number~\cite{Schatzki2024-il, Caro2022-cf}, the number of unitaries required to cover the entire unitary subspace that the QNN can express within a margin of error $\epsilon$.
A large $\epsilon$-covering number indicates a large model expressivity (small bias), leading to a high generalization error (large variance), which can be understood as the bias-variance tradeoff.
The equivariance constrains the expressivity and lowers the covering number, leading to high generalization.

\subsection{Measurement efficiency} \label{sec: expectation_value}

In addition to improving trainability and generalization, our model exhibits higher measurement efficiency than the conventional randomized QCNN.
This high measurement efficiency comes from the effective parallelization of QCNNs.
To show this, let us examine the randomized model that corresponds directly to the equivariant sp-QCNN characterized by $Q_i^{(\ell)}$ and $V_i^{(\ell)}$.
The randomized QCNN randomly selects which branch to remain in each pooling layer for every measurement shot [Fig.~\ref{fig: illust} (b)].
Therefore, the $\ell$th convolutional layer classically mixes the input $\rho$ as $\rho \to \sum_{i=1}^{s_\ell} V^{(\ell)}_i \rho V^{(\ell)\dag}_i/s_\ell$.
The expectation value of a local observable $O=\sum_{i=1}^{n} O_i/n$ ($O_i$ an observable acting only on the $i$th qubit) in the randomized QCNN is defined as the classical summation of the expectation values for several circuits:
\begin{align}
    \braket{O}_\text{RD} = \frac{1}{n} \sum_{i=1}^{n}  \tr\left[ B_i \rho B_i^\dag O_i \right],
\end{align}
where $\braket{\cdot}_\text{RD}$ is the expectation value in the randomized QCNN, and $B_i=V^{(L)}_{i_L} \cdots V^{(1)}_{i_1}$ denotes the executed circuit of the randomized QCNN involving the $i$th qubit.
Here, $B_i$ can also be viewed as the subcircuit (or backward lightcone) of sp-QCNN associated with the $i$th qubit.
In other words, the randomized QCNN chooses one of the subcircuits $B_i$ randomly for every measurement shot, taking an average over the outputs of all the subcircuits.

In the absence of statistical errors in measurements, the sp-QCNN produces the same result as the randomized model.
One can show this by
\begin{align}
    \braket{O}_\text{sp}
    &=\tr(U\rho U^\dag O) \notag \\
    &= \frac{1}{n} \sum_{i=1}^{n} \tr(U\rho U^\dag O_i) \notag \\
    &= \frac{1}{n} \sum_{i=1}^{n} \tr( B_i\rho B_i^\dag O_i) \notag \\
    &= \braket{O}_\text{RD}, \label{eq: Osp=ORD}
\end{align}
where $\braket{\cdot}_\text{sp}$ is the expectation value in the sp-QCNN.
In the third line, we have used $U^\dag O_i U = B_i^\dag O_i B_i$, which is derived from the hierarchical circuit structure and the locality of $O_i$.
This result suggests that the sp-QCNN can coherently execute all the subcircuits in parallel, leading to higher measurement efficiency than the randomized model.
Note that all randomized models do not have corresponding sp-QCNNs because the branches are prohibited from overlapping in sp-QCNNs [i.e., $Q^{(\ell)}_i \cap Q^{(\ell)}_j = \varnothing$ for $i\neq j$].

The coherent parallelization of sp-QCNN enables us to obtain $n$ times more measurement outcomes per circuit execution than the randomized model, potentially leading to the $\mO(n)$ times improvement of measurement efficiency.
This high measurement efficiency reduces the number of measurement shots required to achieve a certain accuracy in estimating the expectation value of an observable.
Note that our model does not exactly improve the measurement efficiency by a factor of $n$ in general because the $n$ measurement outcomes are correlated to each other due to the quantum entanglement of the output state.
For example, when the output state is the GHZ state $\ket{\psi}=(\ket{00\cdots}+\ket{11\cdots})/\sqrt{2}$, the sp-QCNN cannot improve the measurement efficiency because the $n$ measurement outcomes are completely correlated and only one-bit information is available every measurement shot.
In contrast, when the output state is random, the sp-QCNN can improve the measurement efficiency by a factor of $\mO(n)$ because the $n$ measurement outcomes are not correlated effectively~\cite{Chinzei2024-nm}.
Given that the learning process in actual QML problems is generally complicated and can be considered approximately random (at least at the beginning of learning when the parameters are randomly initialized), we expect the sp-QCNN to be resource-efficient.
In Sec.~\ref{sec: demo}, we will verify that the sp-QCNN can considerably improve the measurement efficiency in a specific classification task.

\subsection{Gradient measurement efficiency} \label{sec: gradient}

The equivariant sp-QCNN also improves the measurement efficiency in the gradient estimation.
In VQAs, solving large-scale problems requires an efficient training algorithm.
Among various training algorithms, the gradient-based optimization method is promising, where the parameters are updated based on the gradient of the loss function $L$ as $\bmt\to\bmt-\eta\nabla L$ ($\eta$ is the learning rate).
However, the gradient estimation needs a high computational cost in quantum computing~\cite{Chinzei2025-pa}.
Therefore, improving the gradient measurement efficiency is crucial for the implementation of large-scale models.

The equivariant sp-QCNN allows us to execute the randomized QCNN in parallel even in the gradient measurement.
To show this, we consider the derivative of the expectation value of a local observable $O=\sum_{i=1}^{n} O_i/n$ in the sp-QCNN [see Eq.~\eqref{eq: Osp=ORD}]:
\begin{align}
    \partial_\mu \braket{O}_\text{sp}
    &= \frac{1}{n}\sum_{i=1}^{n} \partial_\mu \tr( B_i\rho B_i^\dag O_i), \label{eq: partial Osp}
\end{align}
where $\partial_\mu$ denotes $\partial/\partial \theta_\mu$.
Suppose that $\theta_\mu$ is a parameter in a branch $Q_\mu \subseteq \Qb$ [e.g., if $\theta_\mu$ is the parameter of $V^{(\ell)}_{i}$, $Q_\mu=Q^{(\ell)}_i$].
Then, we have $\partial_\mu \tr( B_i\rho B_i^\dag O_i)=0$ for $i\notin Q_\mu$ in Eq.~\eqref{eq: partial Osp} because $B_i$ does not depend on $\theta_\mu$ due to the hierarchical structure.
Hence, the derivative is written as 
\begin{align}
    \partial_\mu \braket{O}_\text{sp}
    &= \frac{1}{n}\sum_{i\in Q_\mu} \partial_\mu \tr( B_i\rho B_i^\dag O_i).
\end{align}
Using the parameter-shift rule~\cite{Mitarai2018-ap, Schuld2019-rr}, we can measure the derivative as 
\begin{align}
    \partial_\mu \braket{O}_\text{sp} 
    &= \frac{1}{n}\sum_{i\in Q_\mu} \left[ \tr( B_{i\mu+}\rho B_{i\mu+}^\dag O_i) - \tr( B_{i\mu-}\rho B_{i\mu-}^\dag O_i) \right], \label{eq: parameter-shift}
\end{align}
where $B_{i\mu\pm}$ is the subcircuit of sp-QCNN in which $\theta_\mu$ is replaced with $\theta_\mu \pm \pi/4$.
The randomized model chooses one of the subcircuits $B_i$ ($i\in Q_\mu$) every circuit execution and measures $O_i$ in the parameter-shifted circuit, thus requiring $2|Q_\mu|$ types of quantum circuits for obtaining the derivative $\partial_\mu \braket{O}_\text{RD} = \partial_\mu \braket{O}_\text{sp}$.

The splitting structure of the sp-QCNN enables us to measure the gradient more efficiently.
This high efficiency stems from two factors.
First, the sp-QCNN can measure all the terms in Eq.~\eqref{eq: parameter-shift} with only two types of quantum circuits (i.e., the $\pm \pi/4$ parameter-shifted circuits) by coherently executing $B_i$ ($i\in Q_\mu$) in parallel, which improves the measurement efficiency by a factor of $\mO(|Q_\mu|)$.
Second, the sp-QCNN can measure different derivatives $\partial_\mu \braket{O}$ and $\partial_\nu \braket{O}$ simultaneously if $\theta_\mu$ and $\theta_\nu$ belong to different branches (i.e., $Q_\mu \cap Q_\nu = \varnothing$).
This is due to the fact that, in the Heisenberg picture, the two observables $B_{i\mu\pm}^\dag O_i B_{i\mu\pm}$ and $B_{j\nu\pm}^\dag O_j B_{j\nu\pm}$ in Eq.~\eqref{eq: parameter-shift} are not overlapped and commute with each other.
Therefore, $B_{i\mu\pm}^\dag O_i B_{i\mu\pm}$ and $B_{j\nu\pm}^\dag O_j B_{j\nu\pm}$ are simultaneously measurable.
This technique allows us to simultaneously measure $s_\ell$ derivatives for $s_\ell$ branches in the $\ell$th convolutional layer, improving the gradient measurement efficiency by a factor of $\mO(s_\ell)$.
Combining these two factors, the sp-QCNN achieves the $\mO(|Q_\mu|)\times\mO(s_\ell)\sim \mO(n)$ times improvement of the gradient measurement efficiency (we have used $|Q_\mu|\sim n/s_\ell$ in the $\ell$th layer).

We note that the discussions in Secs.~\ref{sec: expectation_value} and \ref{sec: gradient} do not rely on the equivariance of sp-QCNNs.
This implies that the high measurement efficiency of sp-QCNNs stems from the splitting structure rather than the equivariance.
The equivariant sp-QCNN integrates the splitting structure with the equivariance to simultaneously achieve high trainability, generalization, and measurement efficiency.

\subsection{Absence of barren plateaus} \label{sec: trainability}

A natural question arises regarding the potential impact of the splitting structure of sp-QCNNs on the high trainability of QCNNs. Specifically, one might wonder if the circuit splitting leads to barren plateaus.
We demonstrate this is not the case: the sp-QCNNs do not suffer from barren plateaus, as in conventional QCNNs. 
For simplicity, we consider a non-equivariant sp-QCNN, denoted by $U(\bmt)$, and add single qubit rotations $R(\bma, \bmb) = \prod_{j=1}^n R_j(\alpha_j, \beta_j)$ to the end of the circuit, where $R_j(\alpha_j, \beta_j)=e^{-i\alpha_j X_j}e^{-i\beta_j Z_j}$ is a single qubit rotation gate.
Therefore, the total unitary of these circuits is given by $U_\text{tot}(\bmt,\bma,\bmb) = R(\bma,\bmb) U(\bmt)$.
Here, we show that $U_\text{tot}(\bmt,\bma,\bmb)$ does not exhibit barren plateaus under an assumption.

We consider a cost function $C(\rho)=\tr[U_\text{tot}\rho U_\text{tot}^\dag O ]$ with $O=\sum_j O_j$, where $O_j$ is a local observable acting only on the $j$th qubit.
The cost function is written as the summation of local cost functions $C(\rho) = \sum_j C_j(\rho)$ using $C_j(\rho) = \tr[ U_\text{tot}\rho U_\text{tot}^\dag O_j ] = \tr[ B_j\rho B_j^\dag O_j ]$, where $B_j$ is the subcircuit of $U_\text{tot}$ associated with the $j$th qubit.
Here, we assume that $C_j(\rho)$ does not exhibit barren plateaus:
\begin{align}
    \text{Var}[C_j(\rho)] \in \Omega\left(\frac{1}{\text{poly}(n)}\right), \label{eq: proof1}
\end{align}
where $\text{Var}[A]$ denotes a variance of $A$ in the parameter space.
This assumption is reasonable because the subcircuit $B_j$ has the same structure as a conventional QCNN, in which the absence of barren plateaus is proved~\cite{Pesah2021-py}.
Note that this assumption holds only if each convolutional layer is sufficiently shallow to avoid causing a barren plateau on its own.

By using $\text{Var}[A+B]=\text{Var}[A]+\text{Var}[B]+2\text{Cov}[A,B]$ and $\text{Cov}[A,B]=E[AB]-E[A]E[B]$, the variance of the total cost function is 
\begin{align}
    \text{Var}[C(\rho)] 
    &= \sum_j \text{Var}[C_j(\rho)] \notag \\
    &+ 2\sum_{j<k} E[C_j(\rho) C_k(\rho)] \notag \\
    &- 2\sum_{j<k} E[C_j(\rho)] E[C_k(\rho)], \label{eq: proof2}
\end{align}
where $\text{Cov}[A,B]$ is the covariance of $A$ and $B$, and $E[A]$ is the expectation value of $A$ in the parameter space.
The first term on the right-hand side is $\Omega(1/\text{poly}(n))$ due to the assumption~\eqref{eq: proof1}, but the negative correlation of $C_j$ and $C_k$ in the second and third terms may lead to the exponentially small variance of $C$ in general.
Below, we show that the second and third terms vanish in our ansatz $U_\text{tot}$, proving the absence of barren plateaus.

We first show that the third term on the right-hand side of Eq.~\eqref{eq: proof2} vanishes for $U_\text{tot}(\bmt,\bma,\bmb)=R(\bma,\bmb)U(\bmt)$.
The expectation value of $C_j$ is written as
\begin{align}
    &E[C_j(\rho)] = \frac{1}{\mathcal{N}}\int d\bmt d\bma d\bmb \, \tr \left[ \rho U_\text{tot}^\dag O_j U_\text{tot} \right], \label{eq: proof3}
\end{align}
where $\mathcal{N}$ is the normalization factor, and the integrals of $\bma$ and $\bmb$ run from $0$ to $\pi$.
We notice that
\begin{align}
    &\int_0^{\pi} d\alpha_j \int_0^{\pi} d\beta_j \, R_j^\dag(\alpha_j,\beta_j) O_j R_j(\alpha_j,\beta_j)  = 0 \label{eq: proof4}
\end{align}
for any $O_j$ expressed as a linear combination of $X_j,Y_j$ and $Z_j$.
Thereby, one can easily show that the integration in Eq.~\eqref{eq: proof3} is zero:
\begin{align}
    &E[C_j(\rho)] = 0. \label{eq: proof_E}
\end{align}
Thus, the third term in Eq.~\eqref{eq: proof2} vanishes.
Similarly, the second term in Eq.~\eqref{eq: proof2} is written as 
\begin{align}
    &E[C_j(\rho) C_k(\rho)] \notag \\
    &= \frac{1}{\mathcal{N}}\int d\bmt d\bma d\bmb \, \tr \left[ \rho U_\text{tot}^\dag O_j U_\text{tot} \right] \tr \left[ \rho U_\text{tot}^\dag O_k U_\text{tot} \right]. \label{eq: proof5}
\end{align}
In this integral, $U_\text{tot}^\dag O_k U_\text{tot} = \cdots R_j^\dag O_k R_j \cdots$ does not depend on $\alpha_j$ and $\beta_j$ because $R_j(\alpha_j,\beta_j)$ commutes with $O_k$ for $j\neq k$.
Therefore, the integration by $\alpha_j$ and $\beta_j$ acts only on $U_\text{tot}^\dag O_j U_\text{tot}$ in Eq.~\eqref{eq: proof5}.
Given Eq.~\eqref{eq: proof4}, the second term in Eq.~\eqref{eq: proof2} vanishes:
\begin{align}
    &E[C_j(\rho) C_k(\rho)] = 0. \label{eq: proof_EE}
\end{align}
Combining Eqs.~\eqref{eq: proof1}, \eqref{eq: proof2}, \eqref{eq: proof_E}, and \eqref{eq: proof_EE}, we obtain
\begin{align}
    \text{Var}[C(\rho)] \in \Omega\left(\frac{1}{\text{poly}(n)}\right),
\end{align}
implying that there are no barren plateaus in sp-QCNNs.

While we have focused on the non-equivariant case, we expect that equivariant sp-QCNNs also do not suffer from barren plateaus.
This is due to the fact that the symmetry constraint reduces the circuit expressivity, increasing the variance of the cost function in general~\cite{Larocca2022-so, Ragone2024-hl, Fontana2024-ky}.
We also remark on the necessity of $R(\bma,\bmb)$.
Although we have assumed $R(\bma,\bmb)$ for the proof, it will not be necessary for the absence of barren plateaus in most cases because the situation where the final single qubit rotations determine the existence of barren plateaus is rather peculiar.
Even if barren plateaus depend on the final single qubit rotations, this does not matter in practice because we can easily avoid barren plateaus only by adding single qubit rotations at the end of the circuit.
In Appendix~\ref{ap: variance}, we confirm these considerations by numerically demonstrating the absence of barren plateaus in an equivariant sp-QCNN without $R(\bma,\bmb)$.

\section{Systematic construction of equivariant sp-QCNN} \label{sec: construction}

This section introduces a systematic method of constructing equivariant sp-QCNNs.
Previous studies have proposed several techniques for designing equivariant QNNs, such as the twirling, the nullspace, and the Choi operator methods~\cite{Meyer2023-vx, Nguyen2024-nt}.
However, as these methods do not assume the constraint of circuit splitting, they are not straightforwardly applicable to spatially equivariant sp-QCNNs that are invariant under spatial operations permuting the qubit positions.
Thus, this work provides an alternative method of constructing equivariant sp-QCNNs, called the subgroup method, especially for spatial symmetries.

This section mainly focuses on spatial symmetry $G$, which is a subgroup of the symmetric group $S_n$.
The symmetric group $S_n$ is defined as the set of all bijective functions from $\Qb$ to $\Qb$, namely all permutations of qubits.
Since $G$ is a subgroup of $S_n$, $G$ also represents the permutations of qubits, such as rotation, inversion, and translation of the lattice on which the qubits are defined.
Based on $S_n$, the action of $g\in G$ on $\Qb$ is defined as $g(q_i)=q_j$ ($q_i,q_j\in \Qb$), and the action on a subset $Q_1\subseteq \Qb$ is also defined as $g(Q_1) = \{g(q)|q\in Q_1\}$.
In addition, the action of $G$ on $\Qb$ naturally leads to a unitary representation $U_g$ on the $n$-qubits quantum system as $U_g\ket{\sigma_1 \cdots \sigma_n} = \ket{\sigma_{g(1)}\cdots \sigma_{g(n)}}$, where $\ket{\sigma_1 \cdots \sigma_n}$ is a computational basis with $\sigma_j=0,1$.
We note that while this section only focuses on spatial symmetries, other types of symmetries (e.g., internal symmetries that do not permute the qubit positions) can be incorporated into our sp-QCNN by appropriately designing the unitary operators on each branch $V^{(\ell)}_i$ with the conventional approaches of equivariant QNNs.

For convenience, we define the following terms:
\begin{dfn}[$G$-equivalence of qubits] \label{def: G-equivalence}
    We say that qubits $q_1,q_2 \in \Qb$ are {\it $G$-equivalent} and denote $q_1\sim q_2$ if and only if there exists $g\in G$ such that $g(q_1) = q_2$.
\end{dfn}
\begin{dfn}[$G$-independence of qubits]
    We say that a subset of qubits $Q_1 \subseteq \Qb$ is {\it $G$-independent} if and only if $\forall q_1,q_2\in Q_1$ ($q_1\neq q_2$) are not $G$-equivalent. 
\end{dfn}
\begin{dfn}[$G$-completeness of qubits]
    We say that a subset of qubits $Q_1 \subseteq \Qb$ is {\it $G$-complete} if and only if, for any $q\in \Qb$, there exists $q_1\in Q_1$  that is $G$-equivalent to $q$.
\end{dfn}

\subsection{Subgroup and coset}

For preliminaries, we briefly review some notions of the group theory used in the subgroup method~\cite{Rotman2012-gi}.
Consider a finite group $G$, where a binary operation $G\times G \to G$ is defined with associativity, an identity element, and inverse elements.
A subset $H \subseteq G$ is called a subgroup of $G$ if $H$ is also a group under the same binary operation as $G$.
In this work, we denote $H\leq G$ if $H$ is a subgroup of $G$. 

Given a subgroup $H$ and an element $g\in G$, the left coset $C^H_g$ is defined as follows:
\begin{align}
    C^H_g = gH = \{gh \,|\, h\in H \}.
\end{align}
Here, $|C^H_g|=|H|$ holds for all $g\in G$.
In what follows, we will refer to left cosets as cosets for simplicity.
The definition of cosets readily leads to the fact that a symmetry operation $g\in G$ maps a coset to another one:
\begin{align}
    g (C^H_{g_1}) = g g_1 H = g_2 H = C^H_{g_2} \label{eq: cosetmap}
\end{align}
with $gg_1=g_2$.
An important property of cosets is that different cosets for $H$ either are identical or have no intersection:
\begin{align}
     C^H_{g_1} = C^H_{g_2} \,\,\,\text{or}\,\,\, C^H_{g_1} \cap C^H_{g_2} = \varnothing \label{eq: coset2}
\end{align}
with $\forall g_1,g_2\in G$.
Therefore, $G$ is decomposed into $s=|G|/|H|$ cosets as 
\begin{align}
    G = \bigsqcup_{i=1}^s C^H_i,
\end{align}
where $C^H_i$ denotes the $i$th coset of $H$ in $G$.
We will use the coset decomposition for systematically constructing equivariant sp-QCNNs.

\subsection{Subgroup method}

Here, we present the subgroup method to systematically construct equivariant sp-QCNNs.
This method involves two steps.
The first step determines the circuit splitting so that it does not break a given symmetry.
Then, as the second step, we design a symmetric unitary operator acting on each branch determined in the first step.
Below, we will describe the outline of the subgroup method.
The details are provided in Appendix~\ref{ap: subgroup method}.

\subsubsection{Circuit splitting}

\begin{figure}[t]
	\centering
	\includegraphics[width=\linewidth]{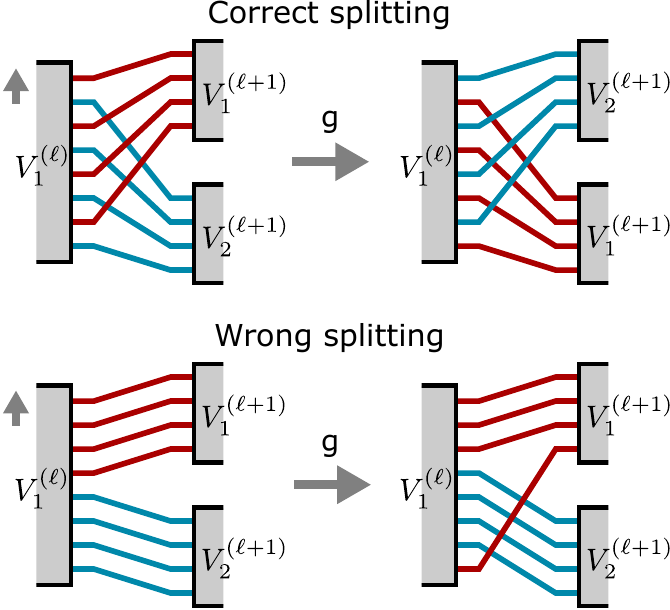}
	\caption{Correct and wrong examples of circuit splitting for translational symmetry $G=\{T^0,T^1,\cdots,T^{n-1}\}$, where $T$ is the translational operation acting as $T(q_i)=q_{i-1}$  ($q_i$ is the $i$th qubit).
The $G$-invariance holds in the correct example: $T(Q_1)=Q_2$ with $Q_1=\{q_1,q_3,q_5,q_7\}$ and $Q_2=\{q_2,q_4,q_6,q_8\}$.
In contrast, the $G$-invariance is broken in the wrong example: $T(Q_1)\neq Q_2$ with $Q_1=\{q_1,q_2,q_3,q_4\}$ and $Q_2=\{q_5,q_6,q_7,q_8\}$.
	}
	\label{fig: circuit splitting}
\end{figure}

\begin{figure*}[t]
	\centering
	\includegraphics[width=\linewidth]{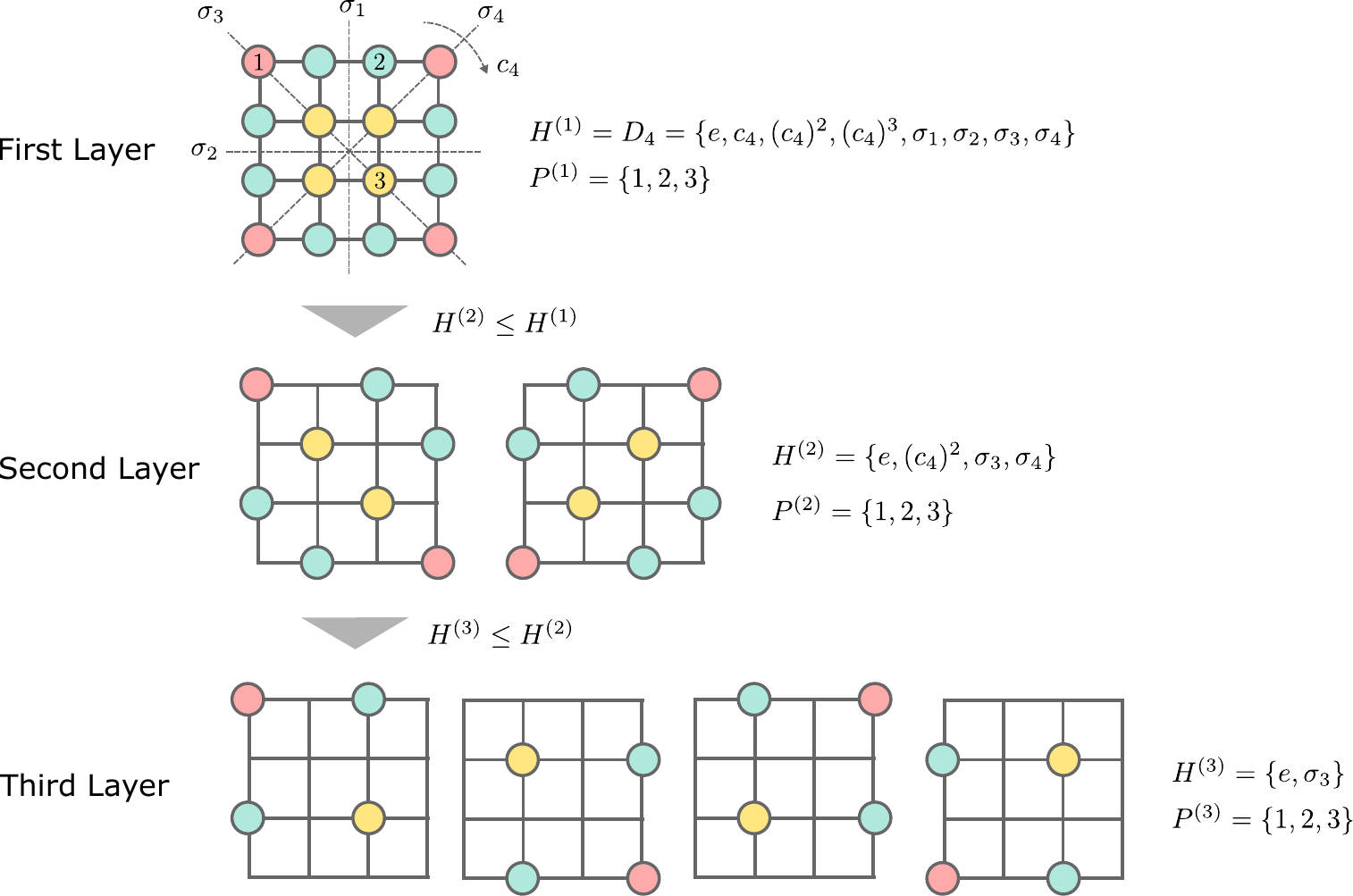}
	\caption{
	   An example of circuit splitting for $D_4=\{e, c_4, (c_4)^2, (c_4)^3, \sigma_1, \sigma_2, \sigma_3, \sigma_4 \}$ symmetry ($e$ is an identity operator, $c_4$ is a rotation by $\pi/2$, and $\sigma_i$'s are inversion operations around each axis).
       Each circle represents a qubit, and qubits with the same color are $G$-equivalent (see Definition~\ref{def: G-equivalence}).
        We can construct the circuit splitting from a subgroup $H^{(\ell)}$ and a qubit subset $P^{(\ell)}$ using the subgroup method.
	}
	\label{fig:pic1}
\end{figure*}

The first step is appropriately splitting the circuit such that $\forall g\in G$ does not change the splitting structure, i.e. the branches $\mQ^{(\ell)}=\{Q_i^{(\ell)}\}_{i=1}^{s_\ell}$.
With the aforementioned Eqs.~\eqref{eq: Qdecomp} and \eqref{eq:cc}, there are three requirements for the $G$-equivariant circuit splitting as follows:
\begin{enumerate}
    \item $G$-invariance of circuit splitting:
    \begin{align}
        \text{$g(\mQ^{(\ell)})=\mQ^{(\ell)}$ for $\forall g\in G$}. \label{eq: G-invariant splitting}
    \end{align}
    \item Branches do not merge:
    \begin{align}
        \text{$\forall i, \exists j$ s.t. $Q_i^{(\ell+1)} \subseteq Q_j^{(\ell)}$}. \label{eq: cond2}
    \end{align}
    \item Branches are a partition of qubits:
    \begin{align}
        &\text{$Q_i^{(\ell)} \cap Q_j^{(\ell)} = \varnothing$ for $i\neq j$}, \label{eq: cond3} \\
        &\text{$\bigcup_i Q_i^{(\ell)} = \Qb$}. \label{eq: cond4}
    \end{align}
\end{enumerate}
In particular, Eq.~\eqref{eq: G-invariant splitting} means that any symmetry operation can permute the branches but never modify the entire splitting structure, as illustrated in Fig.~\ref{fig: circuit splitting}, which is essential for the equivariance.

The subgroup method gives a systematic way of circuit splitting to satisfy Eqs.~\eqref{eq: G-invariant splitting}--\eqref{eq: cond4}.
The key idea is as follows.
We consider a set of subgroups 
\begin{align}
\mc{H}^{(\ell)}=\{H_\lambda^{(\ell)}\}_{\lambda=1}^{\Lambda_\ell}
\end{align}
and a set of qubit subsets
\begin{align}
\mc{P}^{(\ell)}=\{P_\lambda^{(\ell)}\}_{\lambda=1}^{\Lambda_\ell},
\end{align}
where $H_\lambda^{(\ell)}$ and $P_\lambda^{(\ell)}$ are a subgroup of $G$ and a subset of $\Qb$, respectively.
An integer $\Lambda_\ell$ denotes the number of subgroups and subsets for the $\ell$th layer.
From these $\mc{H}^{(\ell)}$ and $\mc{P}^{(\ell)}$, we define the following branch:
\begin{align}
Q_{\lambda,i}^{(\ell)}
&=C^{H_\lambda^{(\ell)}}_i(P_\lambda^{(\ell)}) \notag \\
&=\left\{g(q) \,\middle|\, g\in C^{H_\lambda^{(\ell)}}_i, q\in P_\lambda^{(\ell)} \right\}, \label{eq: Qlambdai}
\end{align}
where $C^{H_\lambda^{(\ell)}}_i$ is the $i$th coset of $H_\lambda^{(\ell)}$ in $G$.
Remarkably, $\mQ^{(\ell)}=\{Q_{\lambda,i}\}_{\lambda,i}$ constructed in this way is $G$-invariant, i.e., it satisfies Eq.~\eqref{eq: G-invariant splitting}.
In fact, one can verify that 
\begin{align}
g(Q_{\lambda,i}^{(\ell)})
&=gC^{H_\lambda^{(\ell)}}_i(P_\lambda^{(\ell)}) 
= C^{H_\lambda^{(\ell)}}_{j} (P_\lambda^{(\ell)})
=Q_{\lambda,j}^{(\ell)},
\end{align}
where we have used Eq.~\eqref{eq: cosetmap} that the symmetry operation $g\in G$ maps a coset to another one.

Based on this idea, we can find $\mc{H}^{(\ell)}$ and $\mc{P}^{(\ell)}$ to satisfy the requirements of Eqs.~\eqref{eq: cond2}--\eqref{eq: cond4} by choosing them layer by layer.
To meet Eq.~\eqref{eq: cond2}, we determine $\mc{H}^{(\ell)}$ and $\mc{P}^{(\ell)}$ as follows:
for $\forall \lambda \in [\Lambda_{\ell+1}]$, there exists $\lambda'  \in [\Lambda_{\ell}]$ such that 
\begin{align}
H_\lambda^{(\ell+1)}\leq H_{\lambda'}^{(\ell)} \quad \text{and} \quad P_\lambda^{(\ell+1)}\subseteq P_{\lambda'}^{(\ell)}. \label{eq: layer by layer}
\end{align}
As proven in Appendix~\ref{ap: cond2}, this condition ensures the requirement of Eq.~\eqref{eq: cond2}.
Moreover, the following conditions are sufficient for the requirements of Eqs.~\eqref{eq: cond3} and \eqref{eq: cond4}:
\begin{align}
    &\text{1. $|G(q)|/|H_\lambda^{(\ell)}(q)|=|G|/|H_\lambda^{(\ell)}|$ for $\forall q\in P_\lambda^{(\ell)}$}, \\
    &\text{2. $P^{(\ell)}$ is $G$-independent}, \\
    &\text{3. $P^{(\ell)}$ is $G$-complete}, \label{eq: last cond}
\end{align}
where we have defined $P^{(\ell)}=\bigsqcup_\lambda P_\lambda^{(\ell)}$. 
The details are provided in Appendix~\ref{ap: cond3}.
To summarize, satisfying all requirements of Eqs.~\eqref{eq: G-invariant splitting}--\eqref{eq: cond4} necessitates careful selection of $\mH^{(\ell)}$ and $\mP^{(\ell)}$ to fulfill the conditions of Eqs.~\eqref{eq: layer by layer}--\eqref{eq: last cond}. 
In Appendix~\ref{ap: brute}, we present a systematic method with proof to determine $\mc{H}^{(\ell)}$ and $\mc{P}^{(\ell)}$ that satisfy these conditions. 
This method proceeds layer by layer from $\ell=L$ to $\ell=1$, ensuring Eqs.~\eqref{eq: layer by layer}--\eqref{eq: last cond} and thus Eqs.~\eqref{eq: G-invariant splitting}--\eqref{eq: cond4} by enumerating all subgroups of $G$.
The computational cost of this method is polynomial in the number of qubits $n$ and, therefore, feasible unless $G$ is too large to be tractable.

As an example of circuit splitting, let us consider a system of 16 qubits defined on a $4\times 4$ square lattice.
This lattice has the square lattice symmetry characterized by $D_4 = \{e,c_4,c_4^2,c_4^3, \sigma_1, \sigma_2,\sigma_3,\sigma_4\}$, where $e$ is an identity operation, $c_4$ is a rotation operation by $\pi/2$, and $\sigma_i$ is an inversion operation for each axis as shown in Fig.~\ref{fig:pic1}.
In the first layer, setting a subgroup $H^{(1)}=D_4$ and a qubit subset $P^{(1)}=\{1, 2, 3\}$ with $\Lambda_1=1$, we have a branch of the first layer as $Q^{(1)}_1=\{1,2,\cdots,16\}$.
In the second layer, considering a subset $H^{(2)}=\{e,c_4^2,\sigma_3,\sigma_4\} \leq H^{(1)}$ and the same qubit subset $P^{(2)}=\{1, 2,3\}$, we obtain the branches of the second layer  $Q_1^{(2)}$ and $Q_2^{(2)}$, as shown in Fig.~\ref{fig:pic1}.
Even in subsequent layers, we can construct the circuit splitting by considering smaller subgroups and qubit subsets sequentially.
Note that the circuit splitting is not unique: a different choice of $\{H_\lambda^{(\ell)}\}$ and $\{P_\lambda^{(\ell)}\}$ results in different circuit splitting.
In particular, although this example only considers the case of $\Lambda=1$ for simplicity, the circuit splitting with $\Lambda >1$ is allowed.

\begin{figure*}[t]
	\centering
	\includegraphics[width=\linewidth]{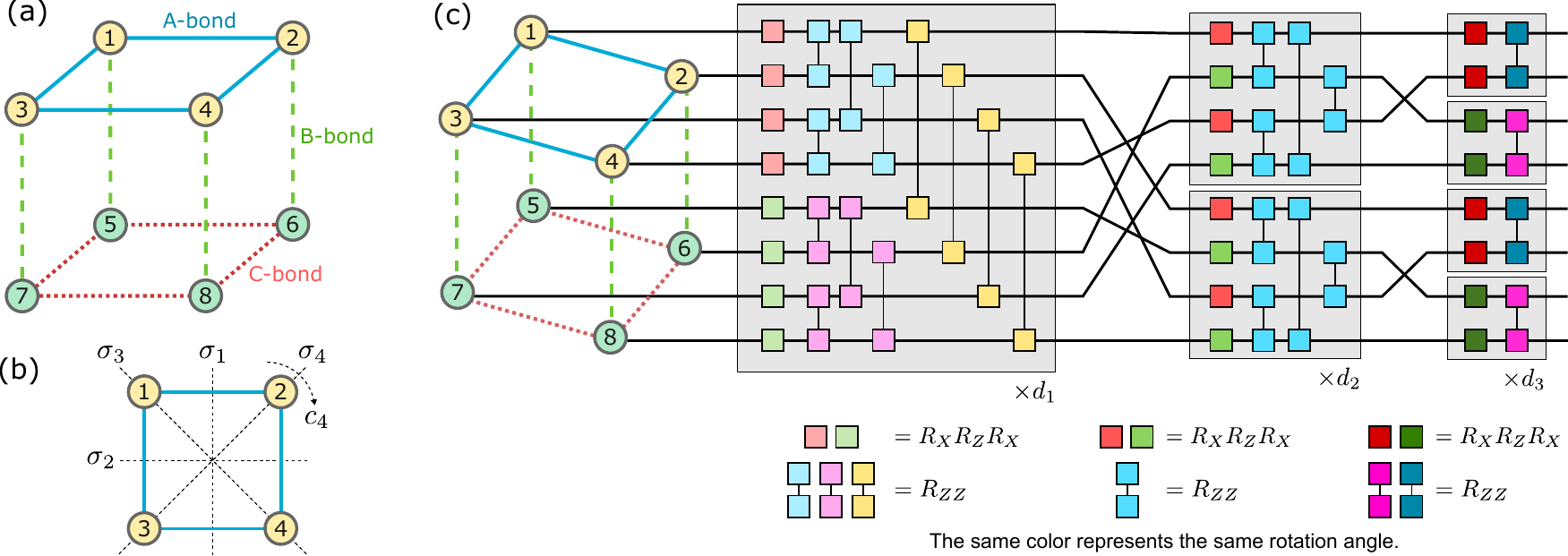}
	\caption{
(a) $2\times2\times2$ cubic lattice that consists of three types of bonds: $A$- (blue), $B$- (green), and $C$- (red) bonds.
(b) This lattice is invariant under the action of $D_4=\{e,c_4,c_4^2, c_4^3,\sigma_1,\sigma_2,\sigma_3,\sigma_4 \}$.
(c) The circuit structure of $D_4$-equivariant sp-QCNN used in this work.
The quantum gates with the same color in each layer share the same rotation angle (e.g., the rotation angles of the four yellow $ZZ$ rotation gates are common).
The details are provided in Appendix~\ref{ap: unitary ansatz}.
	}
	\label{fig: demo_lattice}
\end{figure*}

\subsubsection{Unitary operators}

After determining the circuit splitting at the $\ell$th layer, we need to design the unitary operator acting on each branch such that Eq.~\eqref{eq:Vell_sym} is satisfied.
However, the conventional methods, such as the twirling and nullspace methods, cannot be straightforwardly applied to design the unitary of the equivariant sp-QCNN because its splitting structure acts as a constraint.
In other words, the conventional methods consider a single unitary $V$ satisfying $U_g V U_g^\dag = V$ and do not assume the splitting structure as $V=\prod_{\lambda,i} V_{\lambda,i}$ ($V_{\lambda,i}$'s are unitaries acting on different subsets of qubits).
Below, we omit the layer index because the unitary operator on each layer can be determined independently.

Our approach to finding unitaries that satisfy Eq.~\eqref{eq:Vell_sym} is to reduce $U_g (\prod_{\lambda,i} V_{\lambda,i}) U_g^\dag = \prod_{\lambda,i} V_{\lambda,i}$ for $g\in G$ to the form of $U_h V_{\lambda,i} U_h^\dag = V_{\lambda,i}$ for $h\in H$ ($H$ is a subgroup of $G$), which can be treated with the conventional methods.
Suppose that $\mc{Q}=\{Q_{\lambda,i}\}$ constructed from $\mc{H}=\{H_\lambda\}$ and $\mc{P}=\{P_\lambda\}$ by the subgroup method is given for the $\ell$th circuit splitting.
Let the first coset of $H_\lambda$ be $C^{H_\lambda}_1=H_\lambda$.
In this method, given that the branch $Q_{\lambda,1}=H_\lambda (P_\lambda)$ is invariant under the action of $h\in H_\lambda$ [i.e., $hH_\lambda(P_\lambda) = H_\lambda(P_\lambda)$], we first design $V_{\lambda,1}$ such that
\begin{align}
    &U_h V_{\lambda,1} U_h^\dag = V_{\lambda,1} \label{eq: Vlambda1}
\end{align}
for $\forall h\in H_\lambda$.
Because this condition involves only $V_{\lambda,1}$, the conventional methods, such as the twirling method, can be applied to construct $V_{\lambda,1}$.
Once $V_{\lambda,1}$ is constructed using such a method, the other unitaries $V_{\lambda,i}$ ($i\neq1$) are determined from $V_{\lambda,1}$ through
\begin{align}
    &U_{g_i} V_{\lambda,1} U_{g_i}^\dag = V_{\lambda,i}, \label{eq: Vlambda2}
\end{align}
where $g_i$ is one of the elements of $C^{H_\lambda}_i$.
Then, the unitary operators $V_{\lambda,i}$ constructed in this way satisfy the following relation (see the next paragraph for the proof):
\begin{align}
    &U_g \left(\prod_{i=1}^{s_\lambda} V_{\lambda,i}\right) U_g^\dag = \prod_{i=1}^{s_\lambda} V_{\lambda,i}, \label{eq: Vlambdai}
\end{align}
for $\forall g\in G$ with $s_\lambda=|G|/|H_\lambda|$.
By designing $\prod_{i=1}^{s_\lambda} V_{\lambda,i}$ for each $\lambda$ in this way, we can construct the $G$-equivariant convolutional layer $V$ as
\begin{align}
    V = \prod_{\lambda=1}^{\Lambda} \prod_{i=1}^{s_\lambda} V_{\lambda,i},
\end{align}
which satisfies
\begin{align}
    U_g V U_g^\dag = V
\end{align}
for $\forall g\in G$.

We prove Eq.~\eqref{eq: Vlambdai} by showing that (i) $U_g$ maps $V_{\lambda,i}$ to another $V_{\lambda,j}$ and that (ii) the map $U_g: \{V_{\lambda,i}\}_i \to \{V_{\lambda,i}\}_i$ is a bijection.
The statement (i) is proved as
\begin{align}
U_g V_{\lambda,i} U_g^\dag 
&= U_{g} U_{g_i} V_{\lambda,1} U_{g_i}^\dag U_{g}^\dag \notag \\ 
&= U_{g g_i} V_{\lambda,1} U_{g g_i}^\dag \notag \\ 
&= U_{g_j h} V_{\lambda,1} U_{g_j h}^\dag \notag \\ 
&= U_{g_j}U_{h} V_{\lambda,1} U_{h}^\dag U_{g_j}^\dag \notag \\ 
&= U_{g_j} V_{\lambda,1} U_{g_j}^\dag \notag \\ 
&= V_{\lambda,j},
\end{align}
where we have used Eqs.~\eqref{eq: Vlambda1}--\eqref{eq: Vlambda2} and $U_{g_1}U_{g_2}=U_{g_1g_2}$ for $\forall g_1,g_2\in G$ (because $U_g$ is the representation of $G$) and defined $g g_i = g_j h$ with $\exists h\in H_\lambda$ and $\exists g_j \in C_j^{H_\lambda}$.
In order to prove the statement (ii), it suffices to show that $U_g: \{V_{\lambda,i}\}_i \to \{V_{\lambda,i}\}_i$ is an injection because $\{V_{\lambda,i}\}_i$ is a finite set.
Here, we prove that by contradiction.
Assume that $U_g V_{\lambda,i} U_g^\dag = U_g V_{\lambda,j} U_g^\dag$ for $i\neq j$.
Then, the assumption readily leads to $V_{\lambda,i}=V_{\lambda,j}$, which contradicts the fact that $V_{\lambda,i}$ and $V_{\lambda,j}$ act on different branches $Q_{\lambda,i}$ and $Q_{\lambda,j}$.
Thus, $U_g$ is a bijection.
These results mean that $U_g$ just permutes $V_{\lambda,i}$'s, proving Eq.~\eqref{eq: Vlambdai}, where we have used $[V_{\lambda,i},V_{\lambda',i'}]=0$.

We show how to construct unitary operators in Fig.~\ref{fig:pic1} for example.
Since the first layer has only one branch, we construct the unitary operator acting on it to satisfy $U_h V_1^{(1)} U_h^\dag = V_1^{(1)}$ for $\forall h\in H^{(1)}$, which is feasible using the conventional method.
In the second layer, we first design $V_1^{(2)}$ acting on $Q_1^{(2)}$ [the left branch in Fig.~\ref{fig:pic1}] such that $U_h V_1^{(2)} U_h^\dag = V_1^{(2)}$ for $\forall h\in H^{(2)}$.
Then, we determine $V_2^{(2)}$ as $U_{g_2} V_1^{(2)} U_{g_2}^\dag = V_2^{(2)}$ with $g_2\in C_2^{H^{(2)}}$ (say $g_2=c_4$).
Similarly, we can construct the unitary operators after the third layer.

\section{Demonstration: noisy quantum data classification} \label{sec: demo}

In this section, we numerically demonstrate the high measurement efficiency and generalization of equivariant sp-QCNNs in a specific classification task of noisy quantum data.

\subsection{Classification of Heisenberg model ground states}

The equivariant sp-QCNN is suitable for solving problems associated with symmetry.
Such problems are often encountered in quantum chemistry, physics, and machine learning, which are the main targets of quantum computing.
In practice, when we know the symmetry of quantum states, dynamics, and data distribution of interest in advance, we can exploit the symmetry to improve the accuracy and efficiency of quantum algorithms.
Meanwhile, the QCNN excels at capturing correlations of quantum data at various length scales due to its hierarchical structure, which is relevant for quantum many-body systems~\cite{Cong2019-ov}.
Combining these advantages with the splitting structure, the equivariant sp-QCNN improves trainability, generalization, and measurement efficiency compared to conventional QCNNs in various tasks.

As a demonstration, we numerically investigate the performance of equivariant sp-QCNNs in the classification task of noisy ground states of Heisenberg models.
These models are concise but can describe diverse, intriguing phenomena, such as magnetic orders, topological orders, and spin liquids; thus, they are suitable for benchmarking~\cite{Wu2024-iy}.
Specifically, we consider spin models on a $2\times2\times2$ cubic lattice, where spins (or qubits) and their Pauli operators $X_j, Y_j$, and $Z_j$ are defined on each lattice site.
This lattice consists of three types of bonds: $A$-, $B$-, and $C$-bonds as shown in Fig.~\ref{fig: demo_lattice} (a).
Here, we consider two types of Heisenberg models, $H_1$ and $H_2$, as 
\begin{align}
    &H_\mu = \sum_{\braket{j,k}} J_{jk}^\mu (X_jX_k + Y_jY_k + Z_jZ_k)
\end{align}
with $\mu=1,2$.
The summation of $\braket{j,k}$ runs over all pairs of nearest neighbor sites $j$ and $k$.
Here, the strength of exchange interaction $J_{jk}^\mu$ depends on the bond types as
\begin{align}
&J_{jk}^\mu = 
\begin{cases}
J_A & \braket{j,k}\in \text{$A$-bond} \\ 
\pm J_B & \braket{j,k}\in \text{$B$-bond} \\ 
J_C & \braket{j,k}\in \text{$C$-bond},
\end{cases} \label{eq: bond}
\end{align}
where $\pm$ corresponds to $\mu=1$ and $2$, respectively.
Let $\ket{\psi_\mu}$ be the ground state of $H_\mu$.
The task here is to classify these two ground states $\ket{\psi_1}$ and $\ket{\psi_2}$.

We assume $\ket{\psi_\mu}$ is disturbed by local noise as 
\begin{align}
    \ket{\psi_\mu(\bmn,\bme)}=\prod_{j=1}^n e^{i\epsilon_j (n^x_j X_j + n^y_j Y_j + n^z_j Z_j)} \ket{\psi_\mu},
\end{align}
where $(n_j^x,n_j^y,n_j^z)$ is a unit vector uniformly sampled from the two-dimensional unit sphere, and $\epsilon_j$ is the rotation angle around $(n_j^x,n_j^y,n_j^z)$ sampled from the normal distribution $\propto\exp(-\epsilon_j^2/2\sigma^2)$ with $\sigma=\gamma\pi/2$ ($\gamma$ is the noise level).
These noisy data cannot be distinguished using any single qubit observable because its expectation value for $\ket{\psi_\mu(\bmn,\bme)}$ is zero due to the spin $SU(2)$ symmetry.
Let $\mathcal{D}_\mu$ be the data distribution for $\ket{\psi_\mu(\bmn,\bme)}$.

In this demonstration, we optimize several machine learning models with training data to classify these noisy quantum states and investigate their classification performances.
Here, $2N_t$ training data are given by
\begin{align}
    \{\ket{\phi_k},y_k\}_{k=1}^{2N_t} 
    =&\{ \ket{\psi_1(\bmn_k,\bme_k)}, y_k=1 \}_{k=1}^{N_t} \notag \\
    &\sqcup\{ \ket{\psi_2(\bmn_k,\bme_k)}, y_k=0 \}_{k=N_t+1}^{2N_t},
\end{align}
where $\ket{\psi_\mu(\bmn_k,\bme_k)}$ is sampled from $\mathcal{D}_\mu$ and $y_k$ is the corresponding label ($y_k=1$ for $\mc{D}_1$ and $y_k=0$ for $\mc{D}_2$).
In the main text, we set $J_A=1.0, J_B=1.5, J_C=1.3$, and $\gamma=0.4$ (see Appendix~\ref{ap: noise resilience} for results with different noise levels $\gamma$).

\begin{figure*}[t]
	\centering
	\includegraphics[width=\linewidth]{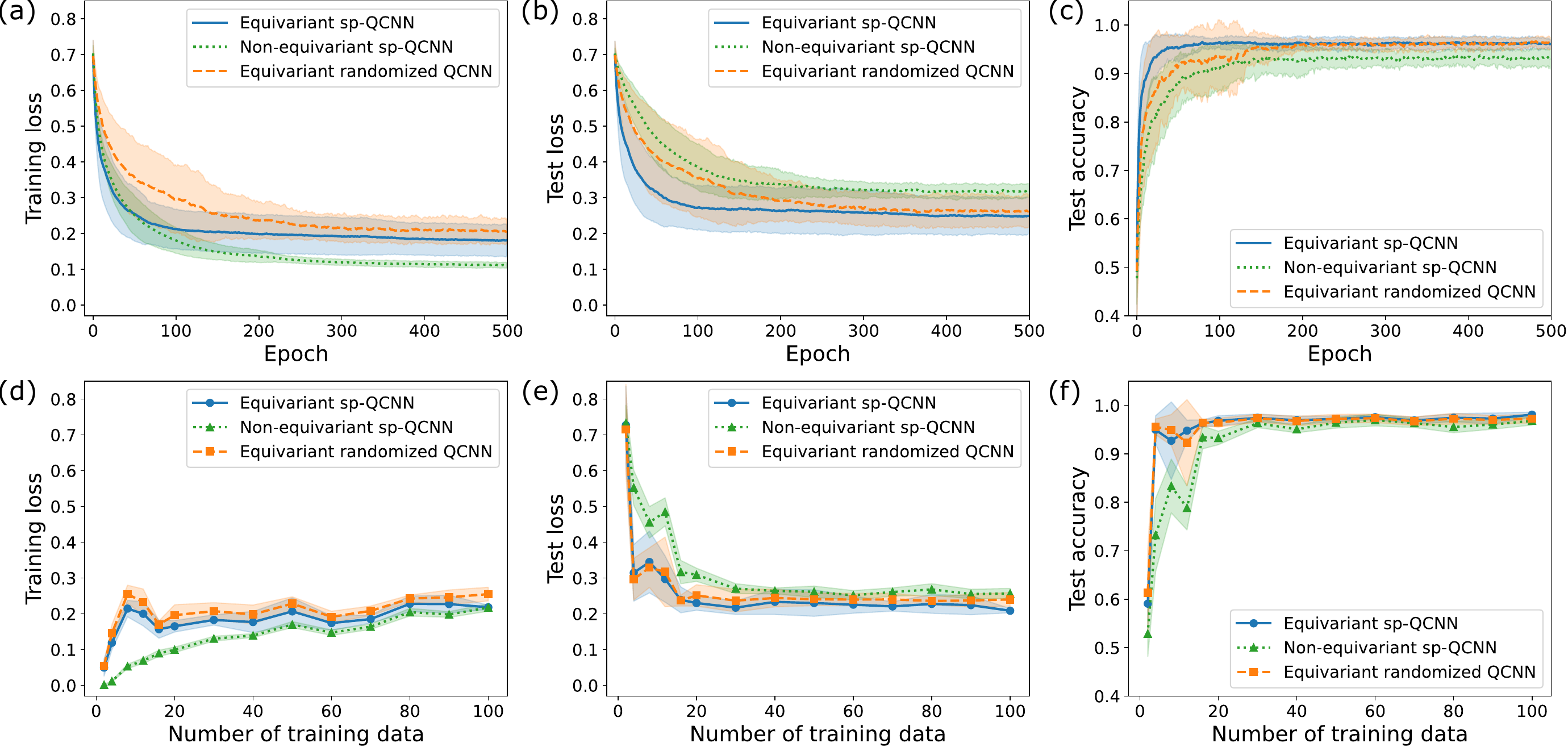}
	\caption{
(a)--(c) Changes in training loss, test loss, and test accuracy during training.
The solid, dotted, and dashed lines denote the results for equivariant sp-QCNN, non-equivariant sp-QCNN, and equivariant randomized QCNN, respectively. 
The shaded areas indicate the standard deviation for the twenty sets of random initial parameters.  
The number of training data is $2N_t=20$.
(d)--(f)  Training loss, test loss, and test accuracy after sufficiently long training processes for various numbers of training data.
The circles, triangles, and squares denote the results for equivariant sp-QCNN, non-equivariant sp-QCNN, and equivariant randomized QCNN, respectively. 
The shaded areas indicate the standard deviation for the twenty sets of random initial parameters.  
	}
	\label{fig: training}
\end{figure*}

\subsection{Machine learning models}

We first clarify the symmetry of the data distribution.
From the bond structure of Eq.~\eqref{eq: bond}, the Hamiltonian $H_\mu$ has the square lattice symmetry characterized by group $D_4$:
\begin{align}
&U_g H_\mu U_g^\dag = H_\mu \quad \forall g\in D_4, \\
&D_4 = \{e,c_4,c_4^2,c_4^3, \sigma_1, \sigma_2,\sigma_3,\sigma_4\},
\end{align}
where $e$ is an identity operation, $c_4$ is a rotation operation by $\pi/2$, and $\sigma_i$ is an inversion operation against each axis.
Owing to this symmetry, if the ground state $\ket{\psi_\mu}$ is not degenerate (i.e., it belongs to a one-dimensional irreducible representation of $D_4$), $\ket{\psi_\mu}$ obeys the following $D_4$-symmetry: 
\begin{align}
    U_g\ket{\psi_\mu} = c_{g\mu} \ket{\psi_\mu} \quad \forall g\in D_4,
\end{align}
where $c_{g\mu}\in\mathbb{C}$ is a phase factor.
This symmetry leads to $U_g\ket{\psi_\mu(\bmn,\bme)}=c_{g\mu}\ket{\psi_\mu(\bmn',\bme')}$, where $(n')_j^\mu=n_{g(j)}^{\mu}$ and $(\epsilon')_j=\epsilon_{g(j)}$.
Since the rotation axes $\bmn$ and angles $\bme$ are sampled randomly and independently, the sampling probabilities of $\ket{\psi_\mu(\bmn,\bme)}$ and $U_g\ket{\psi_\mu(\bmn,\bme)}$ are equivalent:
\begin{align}
    \text{Prob}[\ket{\psi_\mu(\bmn,\bme)}\sim \mathcal{D}_\mu] = \text{Prob}[U_g\ket{\psi_\mu(\bmn,\bme)}\sim \mathcal{D}_\mu] \label{eq: prob equivalence}
\end{align}
for $\forall \bmn\in \mathbb{R}^{3n}, \forall \bme\in \mathbb{R}^n, \forall g\in D_4$.
In other words, given a quantum data $\ket{\psi}$ sampled from $\mc{D}_1$ or $\mc{D}_2$, the probability for $\ket{\psi}$ to have been sampled from $\mc{D}_1$ ($\mc{D}_2)$ is equal to that for $U_g\ket{\psi}$ to have been sampled from $\mc{D}_1$ ($\mc{D}_2)$.
Therefore, machine learning models should be designed to satisfy this symmetry condition of Eq.~\eqref{eq: prob equivalence}.

We consider three types of QCNNs: equivariant and non-equivariant sp-QCNNs and an equivariant randomized QCNN.
Figure~\ref{fig: demo_lattice} (c) depicts the unitary circuit of the equivariant sp-QCNN, which consists of local $R_XR_ZR_X$ and $R_{ZZ}$ rotations (this circuit can be constructed using the subgroup method).
In the equivariant sp-QCNN, some rotation gates share the parameter values with other rotation gates to ensure $D_4$-equivariance.
The non-equivariant sp-QCNN used in this work has the same circuit structure as the equivariant sp-QCNN, but the parameter sharing is not imposed.
In the randomized QCNN, one of the eight subcircuits in the equivariant sp-QCNN (corresponding to an output qubit) is randomly chosen for every measurement shot.
The unitary circuits $U(\bmt)$ of equivariant sp-QCNN and randomized QCNN satisfy $D_4$-symmetry as $[U(\bmt),U_g]=0$ for $\forall g \in D_4$.
The detailed description of the unitary circuit is provided in Appendix~\ref{ap: unitary ansatz}.
In this work, we set the depth of each convolutional layer as $d_1=d_2=d_3=3$, where the total number of independent parameters is $9d_1+7d_2+8d_3=72$ for the equivariant sp-QCNN and randomized model and $36d_1+32d_2+28d_3=288$ for the non-equivariant sp-QCNN.

We employ logistic regression to classify the noisy data.
To this end, after applying $U(\bmt)$ to an input state $\ket{\phi}$, we measure the expectation value of a $D_4$-symmetric observable $O=\sum_j X_j$ ($[O,U_g]=0$ $\forall g\in D_4$), defining the probabilities that $\rho$ belongs to $\mathcal{D}_1$ and $\mathcal{D}_2$ as
\begin{align}
    &p_1(\bmt,\ket{\phi}) = \frac{1}{1 + \exp\left[-\braket{O}\right]}, \\
    &p_2(\bmt,\ket{\phi}) = \frac{1}{1 + \exp\left[\braket{O}\right]},
\end{align}
with $\braket{O}=\braket{\phi | U^\dag(\bmt) O U(\bmt) | \phi}$.
Note that $p_1(\bmt,\ket{\phi})$ and $p_2(\bmt,\ket{\phi})$ satisfy the positivity and the conservation of probabilities: $p_1(\bmt,\ket{\phi}), p_2(\bmt,\ket{\phi})\geq0$ and $p_1(\bmt,\ket{\phi})+p_2(\bmt,\ket{\phi})=1$.
In the equivariant sp-QCNN and randomized QCNN, the $D_4$-symmetry of circuit and observable ensures 
\begin{align}
    p_1(\bmt, \ket{\phi}) = p_1(\bmt, U_g\ket{\phi}), \\
    p_2(\bmt, \ket{\phi}) = p_2(\bmt, U_g\ket{\phi}),
\end{align}
which can improve the trainability and generalization in the classification task.

To train these QCNN models, we consider the following cross entropy as a loss function:
\begin{align}
    L(\bmt) 
    = -\frac{1}{2N_t} \sum_{k=1}^{2N_t} \big[ &y_k \log p_1(\bmt,\ket{\phi_k}) \notag \\
    &+ (1-y_k) \log p_2(\bmt,\ket{\phi_k}) \big].
\end{align}
We optimize this loss function using the Adam algorithm~\cite{Kingma2014-db} with the parameter-shift rule.
The hyper-parameter values used in this work are initial learning rate $=10^{-3}$, $\beta_1=0.9$, $\beta_2=0.999$, and $\eta=10^{-8}$.
We also employ the stochastic gradient descent~\cite{Robbins1951-ql}, where only one training data is used to measure the gradient at each iteration.

In this numerical demonstration, we assume that only a few measurement shots are available for the gradient estimation due to limited computational resources. 
Although the equivariant and non-equivariant sp-QCNNs can simultaneously measure the gradient components for different branches as discussed in Sec.~\ref{sec: gradient}, we match the total number of measurement shots per iteration for the three models to evaluate the measurement efficiency fairly.
The total number of shots for measuring the gradient with the parameter-shift rule is $(2N_\text{gate}+1)N_\text{shot}$, where $N_\text{gate}$ and $N_\text{shot}$ are the number of rotation gates in the circuit and the number of shots per circuit, respectively.
We set $N_\text{gate}=144$ and $N_\text{shot}=5$.
The numerical simulations in this work use Qulacs, an open-source quantum circuit simulator~\cite{Suzuki2021-uh}.

\begin{figure}[t]
	\centering
	\includegraphics[width=0.8\linewidth]{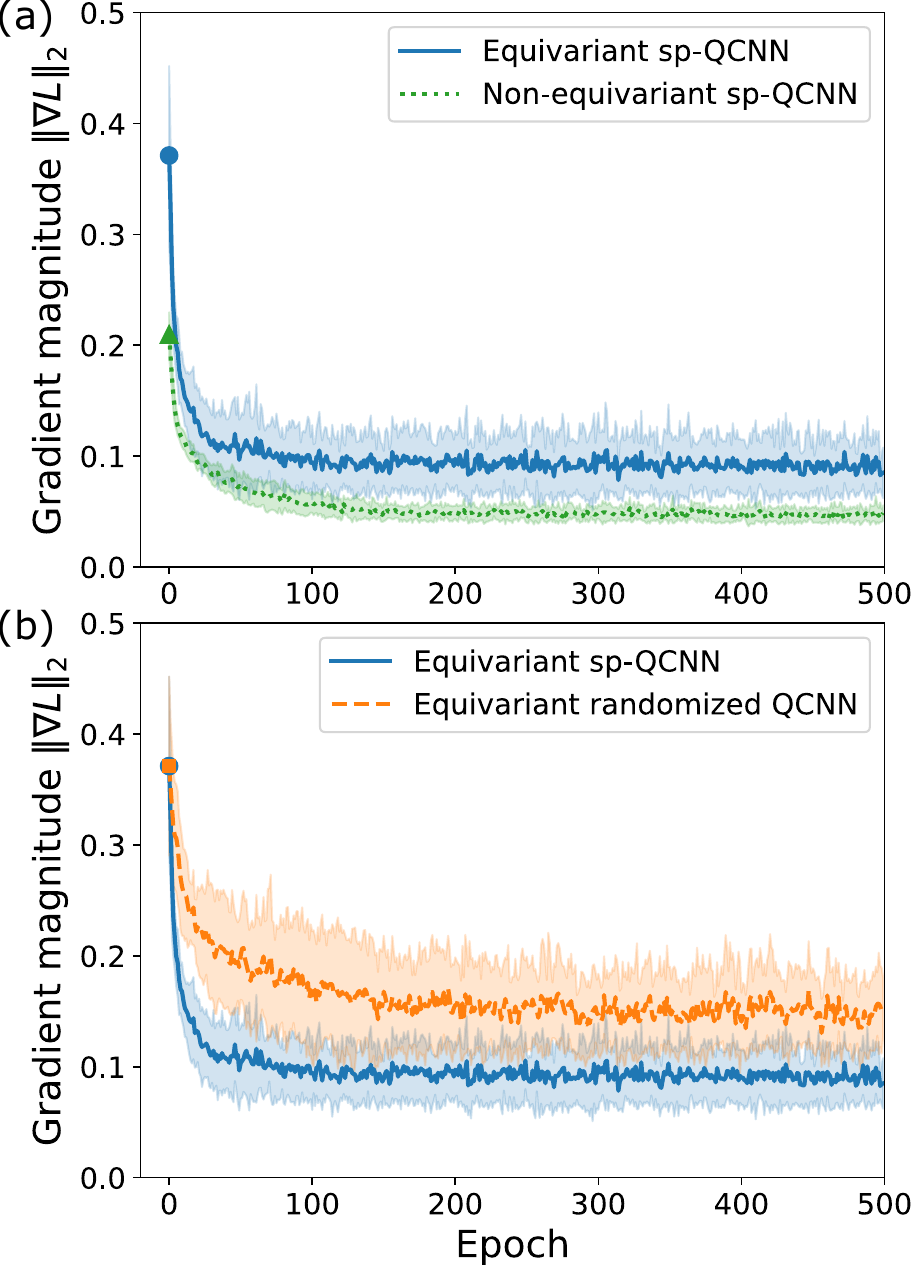}
	\caption{
Gradient magnitudes during training for (a) equivariant and non-equivariant sp-QCNNs and (b) equivariant sp-QCNN and randomized QCNN.
The solid, dotted, and dashed lines represent the mean of results for the twenty sets of random initial parameters, while the shaded areas indicate the standard deviation.  
The circle, triangle, and square denote the initial values of gradient magnitudes for each model.
The training setup is the same as in Figs.~\ref{fig: training} (a)--(c), and the gradient magnitude is calculated for $2N_t=20$ training data.
	}
	\label{fig: grad_mag}
\end{figure}

\begin{figure}[t]
	\centering
	\includegraphics[width=0.8\linewidth]{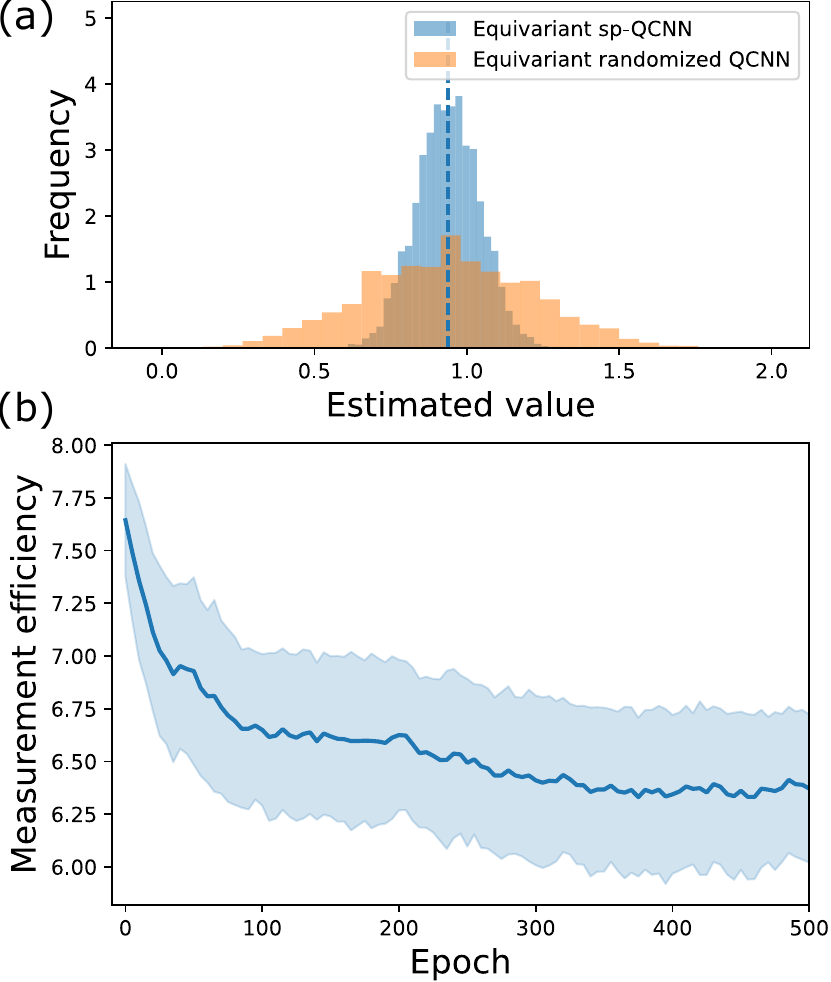}
	\caption{
(a) Histograms for the estimated expectation value of $O=\sum_j X_j$ in the equivariant sp-QCNN and the randomized model at 500 epoch.
The dashed line is the exact expectation value.
We estimate the expectation value with 100 measurement shots, repeating this $T=10,000$ times.
Then, letting $o_1,o_2,\cdots,o_{T}$ be the estimated expectation value for each trial, the histogram shows the frequency of estimated values in $o_1,o_2,\cdots,o_{T}$.
(b) Change in the measurement efficiency of the equivariant sp-QCNN compared to the randomized model during training.
The solid line and shaded area denote the mean and standard deviation for 20 random initial parameter sets.
The measurement efficiency is defined as the ratio between the variances of the estimated expectation values.
	}
	\label{fig: result_eff}
\end{figure}

\subsection{Results}

The first advantage of the equivariant sp-QCNN is the high generalization performance coming from the equivariance.
Figures~\ref{fig: training} (a)--(c) show the changes in the training and test losses and the test accuracy during training.
While the non-equivariant sp-QCNN (green dotted lines) achieves the lowest training loss among the three models, its test loss and accuracy are worse than those of the other models.
This indicates that the non-equivariant sp-QCNN is overfitting for the training dataset due to its excessive expressivity.
In contrast, the equivariant sp-QCNN and randomized QCNN show better test loss and accuracy than the non-equivariant sp-QCNN in Figs.~\ref{fig: training} (b) and (c).
The high generalization of the equivariant models can also be observed in Figs.~\ref{fig: training} (d)--(f), which show the losses and the accuracy after sufficiently long training processes for various numbers of training data.
These results verify that the equivariant models can achieve high test accuracy (or low test loss) with fewer training data than the non-equivariant sp-QCNN. 
For example, both the equivariant models can achieve $90\%$ test accuracy with only 4 training data, whereas the non-equivariant one requires 16 training data for the same accuracy.
These results suggest the significance of equivariance, especially when the number of training data is limited.

In addition, the equivariance can enhance the gradient magnitude of sp-QCNN.
Figure~\ref{fig: grad_mag} (a) shows the changes in the gradient magnitudes of the training loss for equivariant and non-equivariant sp-QCNNs.
We observe that the gradient magnitude of the equivariant sp-QCNN is larger than that of the non-equivariant one, even beyond the beginning of the training process, where the circuit is randomly initialized.
This indicates that the equivariance restricts the expressivity of the quantum circuit, resulting in a larger scale of gradient magnitudes over the loss landscape~\cite{Ragone2024-hl, Fontana2024-ky}.
This enhancement mitigates the vanishing gradient issue, potentially facilitating a smoother and more efficient training process for large-scale problems (see also Appendix~\ref{ap: variance} for the enhanced variance in the cost function due to equivariance).

Another advantage of the equivariant sp-QCNN is the high measurement efficiency stemming from its splitting structure.
In Figs.~\ref{fig: training} (a)--(c), the equivariant sp-QCNN (blue solid lines) displays faster convergences of training loss, test loss, and test accuracy than the randomized QCNN (orange dashed lines).
Since the number of measurement shots used for each epoch is consistent across all the models, this result indicates that our model can be trained with fewer measurement shots.
Furthermore, in Fig.~\ref{fig: grad_mag} (b), the gradient magnitude of the equivariant sp-QCNN converges to a smaller value more quickly compared to the randomized model.
These results support the theoretical analysis for the high measurement efficiency of the equivariant sp-QCNNs.
In the randomized model, the large statistical errors in gradient estimation hinder stable and fast optimization, slowing down the training process (and resulting in the convergence to a larger gradient magnitude).
Conversely, the high measurement efficiency of the equivariant sp-QCNN suppresses the statistical errors to stabilize the optimization and accelerate the training process.

We quantify the measurement efficiency of the equivariant sp-QCNN.
In quantum computing, statistical errors in estimating expectation values are inevitable due to finite measurement resources.
The equivariant sp-QCNN can suppress the statistical errors compared to the randomized model because the splitting structure enables us to obtain $n$ times more measurement outcomes.
In general, the variance of the estimated expectation value decays as $\mO(1/N_\text{shot})$ according to the central limit theorem. 
Therefore, we here define the relative measurement efficiency as $r=v_\text{rand}/v_\text{sp}$, where $v_\text{sp}$ and $v_\text{rand}$ are the variances of the estimated expectation values in the equivariant sp-QCNN and the randomized model, respectively [see Fig.~\ref{fig: result_eff} (a) for example].
This quantification indicates that the equivariant sp-QCNN can achieve the same accuracy as the randomized model only with $1/r$ times fewer measurement shots.
Figure~\ref{fig: result_eff} (b) shows the change in the relative measurement efficiency during training.
In this numerical experiment, we assume that the equivariant sp-QCNN and the randomized model have the same parameters $\bm{\theta}$ to evaluate the relative measurement efficiency between them, where the parameters are optimized through training the equivariant sp-QCNN.
We observe that the measurement efficiency remains high during training: it begins at about 7.7 and converges to 6.3.
That is, the equivariant sp-QCNN can ideally reduce the required measurement shots by at least $1/6.3$ times in this problem.
Note that this calculation of $n=8$ alone is not sufficient to prove that the equivariant sp-QCNN can improve the measurement efficiency by a factor of $\mO(n)$.
Nevertheless, a previous study numerically verified the $\mO(n)$ times improvement for a translationally equivariant sp-QCNN~\cite{Chinzei2024-nm}, suggesting that, with our numerical results, similar improvements are available even for other symmetries.

\section{Conclusions} \label{sec: conclusions}

In this work, we have proposed equivariant sp-QCNNs, an efficient framework that integrates circuit splitting with equivariance for general symmetries.
By maximally leveraging the qubit resource via circuit splitting, the equivariant sp-QCNN ideally improves the measurement efficiency by a factor of $\mO(n)$ compared to the conventional equivariant QCNN.
Furthermore, the equivariance, along with QCNN's intrinsic nature, leads to high trainability and generalization performance.
We have introduced a group-theoretical method of constructing the equivariant sp-QCNN for spatial symmetries, establishing the basis for general model design.
The numerical experiment for the specific classification task has demonstrated that our model outperforms the conventional equivariant and non-equivariant QCNNs in terms of the resources required for training, highlighting the effectiveness of the equivariant sp-QCNN.

As a future research direction, applying the equivariant sp-QCNN to quantum many-body problems is intriguing~\cite{Wu2024-iy}.
It is known that the conventional QCNN has a similar structure to the multiscale entanglement renormalization ansatz (MERA), a representative tensor network model describing one-dimensional critical quantum systems~\cite{Vidal2007-yf}.
On the other hand, the sp-QCNN has a similar structure to the branching MERA, which is known as a good ansatz for describing higher-dimensional quantum critical systems~\cite{Evenbly2014-dk}.
Thus, the sp-QCNN would be suitable for representing higher dimensional critical phenomena in various quantum algorithms, such as variational quantum eigensolvers~\cite{Peruzzo2014-oj}.
This may shed light on unsolved problems in physics, such as identifying the phase diagrams of the Hubbard model~\cite{Arovas2022-ed} and the kagome antiferromagnetic Heisenberg model~\cite{Yan2011-df}.

Finally, we discuss the potential quantum advantages of our model in terms of classical simulability.
A critical conjecture on QNNs has recently been raised, stating that provably barren plateau-free models are classically simulable~\cite{Cerezo2023-hz}.
This conjecture may prevent the exponential quantum advantages of many QNN models.
In particular, a recent paper~\cite{Bermejo2024-dg} theoretically and numerically demonstrated the classical simulability of QCNNs, including our sp-QCNN, for locally easy datasets (i.e., datasets that can be classified by only local information).
Whereas these results appear to show the absence of quantum advantages by QCNNs, these papers also mention some caveats in the conjecture.
Here, we highlight some of them as examples.
First, it remains unclear whether QCNNs can show quantum advantages for classically nontrivial (i.e., not locally easy) datasets.
While such datasets seem to have strong entanglement and cause barren plateaus~\cite{Ortiz-Marrero2021-lx}, some approaches, such as warm start~\cite{Dborin2021-pp, Mele2022-ek, Grimsley2023-gq, Rudolph2023-lz, Puig-i-Valls2024-bo} and curriculum learning~\cite{Tran2025-qh, Recio-Armengol2024-qv}, may avoid the poor trainability issue.
Second, although the papers mainly show the absence of {\it exponential} quantum speedup by QCNNs, the possibility of {\it polynomial} quantum speedup has not been precluded.
In other words, there is no guarantee that the classical simulation process for training QCNNs is faster than executing the original quantum algorithm on compatible hardware~\cite{Lerch2024-cv}.
Our research presents a superior method for accelerating the training of QCNNs on a real quantum computer.
Third, even if the quantum advantages of QCNNs are absent in the training process, QCNNs can be useful in the inference stage, where one
trains a QCNN on a classical computer and then uses the classically optimized parameters to implement a QCNN on a quantum computer for the inference of new data. 
This approach can avoid the need to obtain classical representations of quantum states from new data instances, potentially reducing quantum resources for inference.
To fairly assess these potential advantages in practical scenarios, the disparity between classical and quantum execution times should be considered. 
For example, the gate operation time on superconducting quantum computers is on the order of tens to hundreds of nanoseconds~\cite{Arute2019-im, Kjaergaard2020-ns}, whereas classical computers typically operate with a clock cycle of less than one nanosecond ($=$1 GHz). 
Furthermore, the overhead associated with quantum error correction additionally prolongs the effective execution time of quantum computers.
Given these negative and positive perspectives, all QNNs, not just QCNNs, require more thorough and comprehensive research related to classical simulability and barren plateaus to realize genuine quantum advantages.

\section*{Acknowledgments}

Fruitful discussions with Yuichi Kamata and Nasa Matsumoto are gratefully acknowledged.

\appendix

\section{Scaling in cost function variance} \label{ap: variance}

\begin{figure*}[t]
	\centering
	\includegraphics[width=\linewidth]{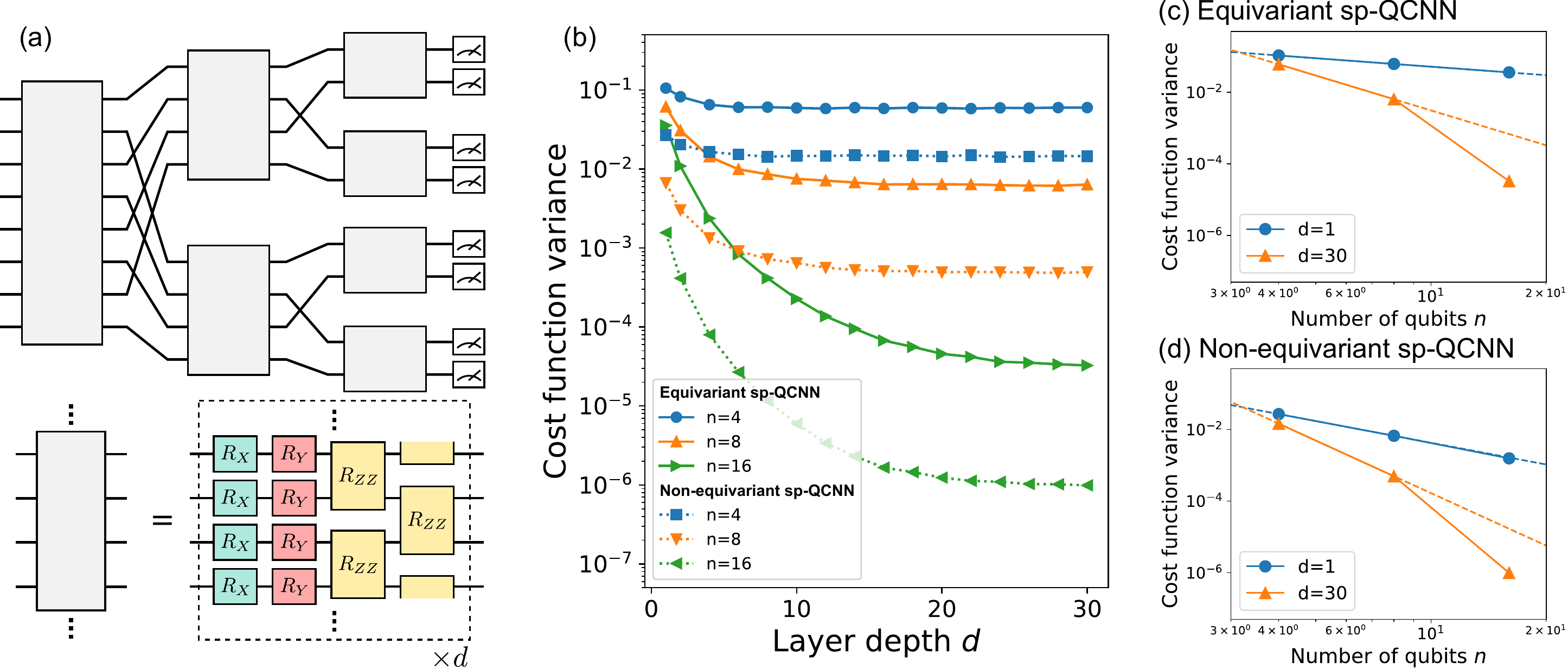}
	\caption{
    (a) Quantum circuit structure used in numerical calculations.
    (b) The variance of the cost function ${\rm Var}_\theta[C(\theta)]$ as a function of the layer depth $d$ for the equivariant and non-equivariant sp-QCNNs with $n=4, 8$, and $16$.
    The variance is calculated by sampling 10000 parameter points from the uniform distribution.
    (c)--(d) The log-log plots of the cost function variance ${\rm Var}_\theta[C(\theta)]$ versus the number of qubits $n$ for $d=1$ and $30$.
    The dashed straight lines connect the first two data points of $n=4$ and $8$, respectively.
	}
	\label{fig: barren_plateau}
\end{figure*}

Here, we numerically demonstrate the absence of barren plateaus in the equivariant sp-QCNN, which is the most crucial factor for scalable QML.
The barren plateau phenomenon is characterized by the exponential decay of cost function variance in the number of qubits $n$: ${\rm Var}_\theta[C(\theta)]\sim 1/b^n$, where $C(\theta)=\tr[U(\theta) \rho U^\dag(\theta) O]$ is the cost function, ${\rm Var}_\theta$ denotes the variance in the parameter space of $\theta$, and $b>1$ is a constant~\cite{McClean2018-qf, Cerezo2021-tq, Ortiz-Marrero2021-lx, Holmes2022-uk}.
If this phenomenon occurs, an exponential amount of measurement resources is required to optimize QNNs, thereby hindering the potential exponential quantum advantages.

To verify the absence of barren plateaus in our model, we consider a translationally equivariant sp-QCNN for simplicity.
The number of qubits calculated is $n=4,8$, and $16$, as the circuit has a consistent architecture when $n=2^k$ ($k$ is an integer). 
The circuit structure of the translationally equivariant sp-QCNN is illustrated in Fig.~\ref{fig: barren_plateau} (a).
The unitary is given by
\begin{align}
    U(\theta) = V^{(L)}(\theta_{L}) \cdots V^{(1)}(\theta_1),
\end{align}
where $L=\log_2 n$, and $V^{(j)}$ is the $j$th convolutional layer.
Each layer is represented as
\begin{align}
    V^{(j)}(\theta_j) = \prod_{k=1}^d \left(  \prod_{\ell=1}^n e^{i\theta_{jk}^{3} Z_\ell Z_{\ell + 2^{j-1}}} \prod_{\ell=1}^n e^{i\theta_{jk}^{2} Y_\ell} \prod_{\ell=1}^n e^{i\theta_{jk}^{1} X_\ell} \right), \label{eq: Tsym spQCNN}
\end{align}
where $d$ is the layer depth.
Since the parameter $\theta_{jk}^{i}$ does not depend on the qubit index $\ell$, this unitary is translationally invariant: $[U(\theta),T]=0$, where $T$ is the translation operation (e.g., $T\ket{100\cdots0}=\ket{010\cdots0}$).
For comparison, we also consider a non-equivariant sp-QCNN, which shares the same circuit architecture as the equivariant one but allows the parameter $\theta_{jk}^{i}$ to vary across different qubits [i.e., $\theta_{jk}^{i}$ is replaced with $\theta_{jk\ell}^{i}$ in Eq.~\eqref{eq: Tsym spQCNN}].
We set the observable to be $O=\sum_{j=1}^n X_j/n$ and the initial state to be $\rho=(\ket{0}\bra{0})^{\otimes n}$.

Figure~\ref{fig: barren_plateau} (b) shows ${\rm Var}_\theta[C(\theta)]$ as a function of the layer depth $d$ for both equivariant and non-equivariant sp-QCNNs with $n=4,8$ and $16$.
We observe that the variance monotonically decreases with $d$ and then converges to a finite value, which becomes smaller as $n$ increases.
Furthermore, ${\rm Var}_\theta[C(\theta)]$ for the equivariant sp-QCNN is consistently larger than that of the non-equivariant one across all $d$ and $n$. 
This suggests that the equivariance restricts the circuit expressivity, thereby leading to a larger variance in the cost function.

We analyze the scaling of the cost function variance with respect to $n$.
Figures~\ref{fig: barren_plateau} (c) and (d) present the log-log plots of ${\rm Var}_\theta[C(\theta)]$ versus $n$ for shallow-layer ($d=1$) and deep-layer ($d=30$) cases, respectively.
For both equivariant and non-equivariant models, the data points for $d=1$ fall on a straight line, which indicates that ${\rm Var}_\theta[C(\theta)]$ exhibits polynomial scaling for the shallow-layer case.
Conversely, the data points for $d=30$ do not fall on a straight line but display an upward convex shape, implying super-polynomial decay in the deep-layer case.
This super-polynomial (typically exponential) scaling is consistent with the results of recent theoretical studies proving that ${\rm Var}_\theta[C(\theta)]$ in the deep-layer limit is inversely proportional to the circuit expressivity, which is exponential in our model~\cite{Ragone2024-hl, Fontana2024-ky}.
To summarize, the sp-QCNNs can avoid the barren plateaus by setting the layer depth to be shallow (more rigorously, a constant independent of $n$).

We offer several caveats.
First, the number of qubits $n$ investigated here is limited by available computational power.
Identifying a more accurate scaling relationship will necessitate larger-scale simulations beyond $n=16$.
Furthermore, although this numerical experiment considered the translationally equivariant sp-QCNN, we expect other types of equivariant sp-QCNN can also avoid barren plateaus.
In fact, even the non-equivariant sp-QCNN, which should exhibit a larger cost function variance than general equivariant sp-QCNNs, also does not show the barren plateaus if the layers are shallow.
Finally, addressing the true scalability of QNNs must also encompass other crucial factors, such as the extensive number of poor local minima~\cite{Anschuetz2022-gb}.
Understanding and resolving these remaining challenges is important for realizing practical QNNs.

\section{Details of circuit splitting} \label{ap: subgroup method}

In this Appendix, we provide the details of circuit spitting in the subgroup method.

\subsection{Sufficient conditions for circuit splitting}

This section provides sufficient conditions for appropriate circuit splitting, which forms the basis for the subgroup method.
Before moving on to the details, we summarize the requirements for the $G$-equivariant circuit splitting $\mQ^{(\ell)}=\{Q_i^{(\ell)}\}_i$ as follows:
\begin{enumerate}
    \item $G$-invariance:
    \begin{align}
        \text{$g(\mQ^{(\ell)})=\mQ^{(\ell)}$ for $\forall g\in G$}. \label{apeq: cond1}
    \end{align}
    \item Branches do not merge:
    \begin{align}
        \text{$\forall i, \exists j$ s.t. $Q_i^{(\ell+1)} \subseteq Q_j^{(\ell)}$}. \label{apeq: cond2}
    \end{align}
    \item Branches are a partition of qubits $\Qb=[n]$:
    \begin{align}
        &\text{$Q_i^{(\ell)} \cap Q_j^{(\ell)} = \varnothing$ for $i\neq j$}, \label{apeq: cond3} \\
        &\text{$\bigcup_i Q_i^{(\ell)} = \Qb$}. \label{apeq: cond4}
    \end{align}
\end{enumerate}
In what follows, we provide sufficient conditions for these three requirements one by one.

For convenience, we recall the following terms regarding the action of $G$ on $\Qb$.
\setcounter{dfn}{0}
\begin{dfn}[$G$-equivalence of qubits]
    We say that qubits $q_1,q_2 \in \Qb$ are {\it $G$-equivalent} and denote $q_1\sim q_2$ if and only if there exists $g\in G$ such that $g(q_1) = q_2$.
\end{dfn}
\begin{dfn}[$G$-independence of qubits]
    We say that a subset of qubits $Q_1 \subseteq \Qb$ is {\it $G$-independent} if and only if $\forall q_1,q_2\in Q_1$ ($q_1\neq q_2$) are not $G$-equivalent. 
\end{dfn}
\begin{dfn}[$G$-completeness of qubits]
    We say that a subset of qubits $Q_1 \subseteq \Qb$ is {\it $G$-complete} if and only if, for any $q\in \Qb$, there exists $q_1\in Q_1$  that is $G$-equivalent to $q$.
\end{dfn}

\subsubsection{Sufficient condition for $G$-invariance}

Here, we show that the subgroup method introduced in the main text ensures the $G$-invariance of Eq.~\eqref{apeq: cond1}.
To this end, we first prove the following theorem, which is the core of this method. 
\begin{thm} \label{thm:split}
Let $H$ be a subgroup of $G$ and $P$ a subset of $\Qb$.
Given the coset decomposition $G=\bigsqcup_{i=1}^{s} C^H_i$, we define
\begin{align}
Q_i=C^H_i(P)=\{g(q) \,|\, g\in C^H_i, q\in P \}.
\end{align}
Then, $\mQ=\{Q_i\}_{i\in [s]}$ is $G$-invariant.
\end{thm}   

\begin{proof}
We prove this theorem by showing that (i) $g\in G$ maps $Q_i\in\mQ$ to another $Q_j\in\mQ$
[i.e., $g(Q_i) = Q_j$] and that (ii) the map $g:\mQ\to\mQ$ is a bijection.

First, we prove the statement (i).
According to the group theory, the action of $g\in G$ changes a coset $C^H_i$ to another coset $C^H_j$, namely $gC^H_i = C^H_j$.
Using this property, we can show the statement (i) as
\begin{align}
    g(Q_i)=g C^H_i(P) = C^H_j(P)=Q_j.
\end{align}

Next, we prove the statement (ii).
Because $\mQ$ is a finite set, it suffices to show that the map $g$ is an injection.
We here give the proof by contradiction. 
Assume $g(Q_i)=g(Q_j)$ for $Q_i\neq Q_j$.
Then, multiplying $g^{-1}$ on both sides, we have $Q_i=Q_j$, which contradicts $Q_i \neq Q_j$ in the assumption.
This shows that $g(Q_i) \neq g(Q_j)$ for $Q_i\neq Q_j$, proving that the map $g$ is an injection and thus a bijection.
\end{proof}

This theorem proves the following corollary, which shows that the circuit splitting constructed by the subgroup method is guaranteed to be $G$-invariant.
\begin{cor} \label{cor:multisplit}
Let $H_\lambda$ be a subgroup of $G$ and $P_\lambda$ a subset of $\Qb$ ($\lambda=1,\cdots,\Lambda$).
Given the coset decompositions $G=\bigsqcup_{i=1}^{s_\lambda} C^{H_\lambda}_i$, we define
\begin{align}
Q_{\lambda,i}=C^{H_\lambda}_i(P_\lambda) = \{g(q) \,|\, g\in C^{H_\lambda}_i, q\in P_\lambda \}.
\end{align}
Then, $\mQ=\{Q_{\lambda,i}\}_{\lambda\in [\Lambda], i\in [s_\lambda]}$ is $G$-invariant.
\end{cor}

\begin{proof}
We consider subsets of $\mQ$ as follows:
\begin{align}
    \mathcal{S}_\lambda \equiv \{Q_{\lambda,i}\}_{i\in [s_\lambda]},
\end{align}
where $\mQ = \bigcup_\lambda \mathcal{S}_\lambda$.
According to Theorem~\ref{thm:split}, $\mathcal{S}_\lambda$ is $G$-invariant.
Therefore, the union of $\mathcal{S}_\lambda$, namely $\mQ$, is also $G$-invariant.
\end{proof}

\subsubsection{Sufficient condition for branches not to merge} \label{ap: cond2}

Here, we give a sufficient condition for branches not to merge in the pooling layer [Eq.~\eqref{apeq: cond2}].

\begin{thm} \label{thm: no merge}
Let $\mathcal{Q}^{(\ell)}=\{Q_{\lambda,i}^{(\ell)}\}$ be $G$-invariant branches constructed by the subgroup method with subgroups $\{H^{(\ell)}_\lambda\}$ and qubit subsets $\{ P^{(\ell)}_\lambda \}$.
If, for $\forall \lambda$, there exists $\lambda'$ such that 
\begin{align}
    H_\lambda^{(\ell+1)} \leq H_{\lambda'}^{(\ell)}, \quad P_\lambda^{(\ell+1)} \subseteq P_{\lambda'}^{(\ell)},
\end{align}
then $\mathcal{Q}^{(\ell)}$ and $\mathcal{Q}^{(\ell+1)}$ satisfy Eq.~\eqref{apeq: cond2}, i.e., the branches do not merge in the pooling layer.
\end{thm} 

\begin{proof}
According to the group theory, if $H_\lambda^{(\ell+1)} \leq H_{\lambda'}^{(\ell)}$, any coset of $H_\lambda^{(\ell+1)}$ is included in a corresponding coset of $H_{\lambda'}^{(\ell)}$ as
\begin{align}
\forall i, \exists j \,\,\,\text{s.t.}\,\,\, C^{H_\lambda^{(\ell+1)}}_i \subseteq C^{H_{\lambda'}^{(\ell)}}_j,
\end{align}
where $C^{H_\lambda^{(\ell+1)}}_i$ and $C^{H_{\lambda'}^{(\ell)}}_j$ are cosets of $H_\lambda^{(\ell+1)}$ and $H_{\lambda'}^{(\ell)}$, respectively.
This leads to an inclusion relation as
\begin{align}
    Q_{\lambda,i}^{(\ell+1)}
    &= C^{H_\lambda^{(\ell+1)}}_i (P_\lambda^{(\ell+1)}) \notag\\ 
    &\subseteq C^{H_{\lambda'}^{(\ell)}}_j (P_\lambda^{(\ell+1)}) \notag\\
    &\subseteq C^{H_{\lambda'}^{(\ell)}}_j (P_{\lambda'}^{(\ell)}) \notag \\ 
    &=Q_{\lambda',j}^{(\ell)},
\end{align}
where we have used $C^{H_\lambda^{(\ell+1)}}_i \subseteq C^{H_{\lambda'}^{(\ell)}}_j$ in the second line and $P_\lambda^{(\ell+1)} \subseteq P_{\lambda'}^{(\ell)}$ in the third line.
This proves Eq~\eqref{apeq: cond2}.
\end{proof}

\subsubsection{Sufficient condition for branches to be a partition of qubits} \label{ap: cond3}

Here, we prove the following Theorem, providing a sufficient condition for branches to be a partition of qubits [Eqs.~\eqref{apeq: cond3} and \eqref{apeq: cond4}]:
\begin{thm} \label{thm: no overlap}
Consider branches $\{Q_{\lambda,i}\}$ defined by $\mH=\{H_\lambda\}$ and $\mP=\{P_\lambda\}$.
Then, the following conditions are sufficient for branches to be a partition of qubits, i.e., to satisfy Eqs.~\eqref{apeq: cond3} and \eqref{apeq: cond4}:
\begin{enumerate}
    \item[(i)] $|G(q)|/|H_\lambda(q)|=s_\lambda$ for $\forall q\in P_\lambda$,
    \item[(ii)] $\bigsqcup_\lambda P_\lambda$ is $G$-independent,
    \item[(iii)] $\bigsqcup_\lambda P_\lambda$ is $G$-complete,
\end{enumerate}
where $s_\lambda=|G|/|H_\lambda|$ is the number of independent cosets.
\end{thm}
\begin{proof}

We first prove that conditions (i) and (ii) are sufficient for Eq.~\eqref{apeq: cond3}.
By condition (ii), we have $G(P_\lambda) \cap G(P_{\lambda'}) = \varnothing$ for $\lambda \neq \lambda'$.
This leads to $Q_{\lambda,i} \cap Q_{\lambda',i'} = \varnothing$ for $\lambda \neq \lambda'$, where we have used $Q_{\lambda,i}\subseteq G(P_\lambda)$ and $Q_{\lambda',i'}\subseteq G(P_{\lambda'})$.
Thus, it suffices to prove $Q_{\lambda,i} \cap Q_{\lambda,i'} = \varnothing$ for $i\neq i'$.

We prove this by contradiction.
Assume that there exist $i\neq i'$ such that $Q_{\lambda,i} \cap Q_{\lambda,i'} \neq \varnothing$.
By definition of coset decomposition $G=\bigsqcup_{i=1}^{s_\lambda} C^{H_\lambda}_i$, we have
\begin{align}
    G(P_\lambda) 
    &= \{g(q) \,|\, g\in G, q\in P_\lambda\} \\
    &= \bigcup_{i=1}^{s_\lambda} \{g(q) \,|\, g\in C^{H_\lambda}_i, q\in P_\lambda\} \\
    &= \bigcup_{i=1}^{s_\lambda} Q_{\lambda,i}. \label{apeq: omake1}
\end{align}
Taking the norm on the leftmost and rightmost sides, the following inequality holds:
\begin{align}
    |G(P_\lambda)| = \left|\bigcup_{i=1}^{s_\lambda} Q_{\lambda,i} \right| < \sum_{i=1}^{s_\lambda} |Q_{\lambda,i}|, \label{apeq: lemma_EG}
\end{align}
where we have used the assumption that $Q_{\lambda,i} \cap Q_{\lambda,i'} \neq \varnothing$ for some $i\neq i'$.
Because of $|Q_{\lambda,i}|=|g_iH_\lambda(P_\lambda)|=|H_\lambda(P_\lambda)|$, the above inequality is reduced to
\begin{align}
    |G(P_\lambda)| < s_\lambda |H_\lambda(P_\lambda)|. 
\end{align}
Given that $P_\lambda$ is $G$-independent by condition (ii), we have
\begin{align}
    &|G(P_\lambda)| = \sum_{q\in P_\lambda} |G(q)|, \\
    &|H_\lambda(P_\lambda)| = \sum_{q\in P_\lambda} |H_\lambda(q)|,
\end{align}
and thus
\begin{align}
    \sum_{q\in P_\lambda} |G(q)| < s_\lambda \sum_{q\in P_\lambda} |H_\lambda(q)|. 
\end{align}
This contradicts condition (i), i.e., $|G(q)|/|H_\lambda(q)|=s_\lambda$ for all $q\in P_\lambda$.
Therefore, the assumption that there exist $i\neq i'$ such that $Q_{\lambda,i} \cap Q_{\lambda,i'} \neq \varnothing$ is incorrect, and conditions (i) and (ii) are sufficient for Eq.~\eqref{apeq: cond3}.

We finally prove that condition (iii) leads to Eq.~\eqref{apeq: cond4}.
From Eq.~\eqref{apeq: omake1}, we have $\bigcup_{\lambda,i} Q_{\lambda,i} = \bigcup_{\lambda} G(P_\lambda)$.
By condition (iii), we also have $\bigcup_{\lambda} G(P_\lambda) = \Qb$, thereby obtaining $\bigcup_{\lambda,i} Q_{\lambda,i} = \Qb$.
This proves that condition (iii) is sufficient for Eq.~\eqref{apeq: cond4}.
\end{proof}

Below, we say that a qubit $q\in \Qb$ is {\it well-behaved} for a subgroup $H$ if $|G(q)|/|H(q)|=|G|/|H|=s$ holds.
Note that confirming whether a qubit $q\in\Qb$ is well-behaved is easy unless $|G|$ is too large to track all elements of $G(q)$.

\begin{figure*}[t]
	\centering
	\includegraphics[width=0.8\linewidth]{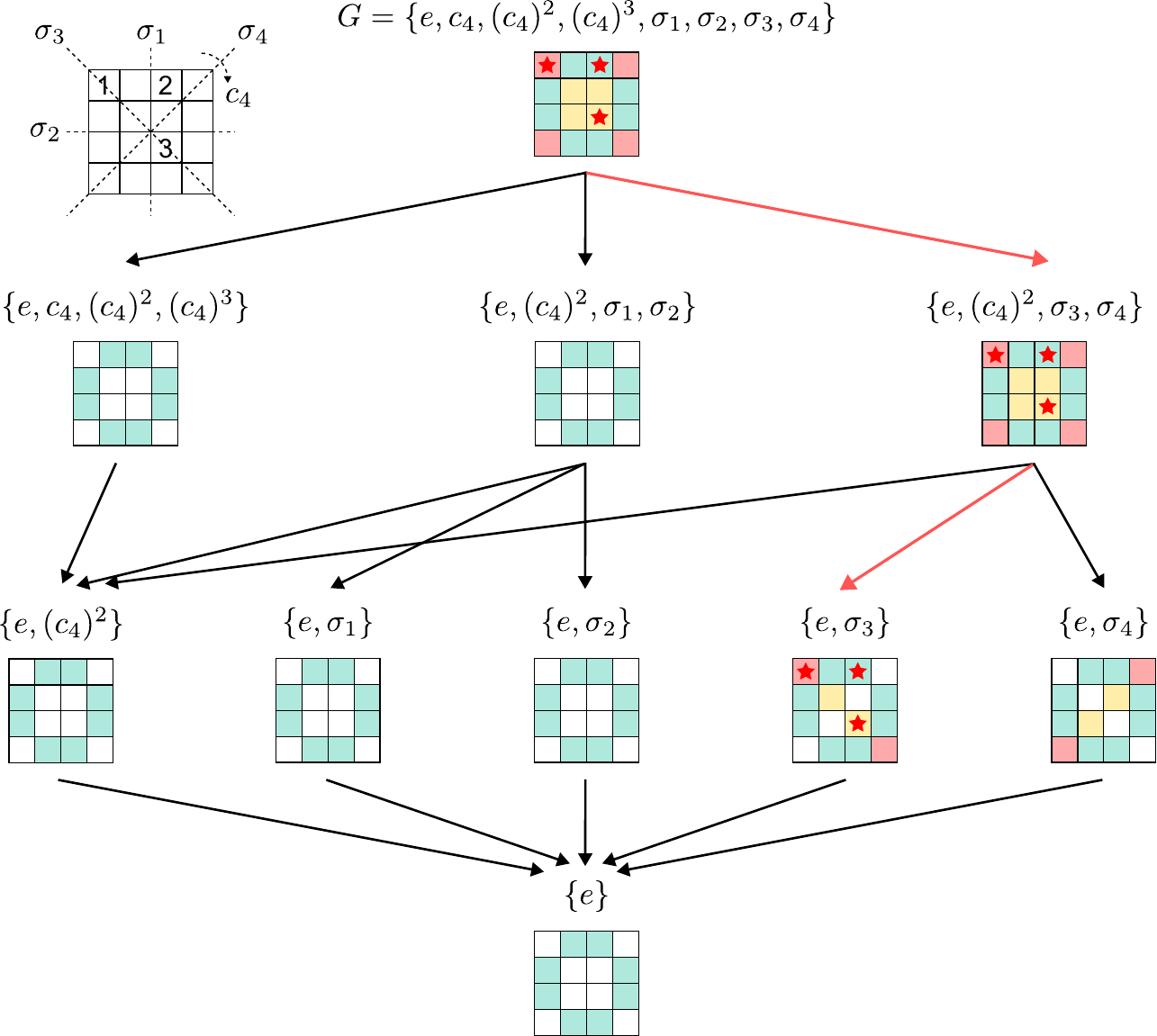}
	\caption{
All subgroups of $D_4$ and well-behaved qubits on the $4\times 4$ lattice.
The qubit represented by a colored (not white) box is well-behaved for the corresponding subgroup.
The same colored (red, yellow, or blue) boxes denote $G$-equivalent qubits.
We can reproduce the circuit splitting of Fig.~\ref{fig:pic1} by choosing the subgroups and qubit subsets marked by the red stars.
	}
	\label{fig: subgroups}
\end{figure*}

\subsection{Systematic method of circuit splitting} \label{ap: brute}


We are ready to present a systematic method of finding the circuit splitting that satisfies Eqs.~\eqref{apeq: cond1}--\eqref{apeq: cond4}.
By Corollary~\ref{cor:multisplit}, the subgroup method allows us to obtain the $G$-invariant branches $\mQ^{(\ell)}=\{Q^{(\ell)}_{\lambda,i}\}$ based on subgroups $\mH^{(\ell)}=\{H^{(\ell)}_\lambda\}$ and qubit subsets $\mP^{(\ell)}=\{P^{(\ell)}_\lambda\}$.
In addition to the $G$-invariance, the branches $\mQ^{(\ell)}$ must satisfy Eqs.~\eqref{apeq: cond2}--\eqref{apeq: cond4}.
According to Theorems~\ref{thm: no merge} and \ref{thm: no overlap}, the following conditions are sufficient:
\begin{align}
    &\text{1. $\forall \lambda, \exists \lambda'$ s.t. $H^{(\ell+1)}_\lambda \leq H^{(\ell)}_{\lambda'}$ and $P^{(\ell+1)}_\lambda \subseteq P^{(\ell)}_{\lambda'}$,} \label{apeq: reduced cond1} \\
    &\text{2. $|G(q)|/|H_\lambda^{(\ell)}(q)|=s_\lambda^{(\ell)}$ for $\forall q\in P_\lambda^{(\ell)}$,} \label{apeq: reduced cond2} \\
    &\text{3. $\bigsqcup_\lambda P_\lambda$ is $G$-independent and $G$-complete,} \label{apeq: reduced cond3}
\end{align}
where $s_\lambda^{(\ell)}=|G|/|H_\lambda^{(\ell)}|$.

In order to find $\mH^{(\ell)}$ and $\mP^{(\ell)}$ satisfying these conditions, we employ a brute-force method.
This method begins with considering all subgroups of $G$ and corresponding well-behaved qubits.
For instance, Fig.~\ref{fig: subgroups} shows all subgroups of $G=D_4$ for a $4\times 4$ qubits array.
The qubit $q\in\Qb$ represented by a colored (not white) box is well-behaved for the corresponding subgroup $H$.
The same colored (red, yellow, or blue) boxes, say $q_1, q_2 \in \Qb$, denote $G$-equivalent qubits that are mapped to each other by an action of $g\in G$ as $g(q_1)=q_2$.
Obtaining this diagram requires only a polynomial computational cost in the number of qubits $n$ and is thus practically possible unless $|G|$ is so large that all subgroups are not available.

Based on this diagram, we first set subgroups and qubit subsets of the final layer $\{ (H_\lambda^{(L)}, P_\lambda^{(L)}) \}_\lambda$.
We choose them such that 
\begin{enumerate}
    \item[] $\bigsqcup_\lambda P_\lambda^{(L)}$ contains one qubit of each color.
\end{enumerate}
The final layer constructed in this way, where $\bigsqcup_\lambda P_\lambda^{(L)}$ contains only colored (i.e., well-behaved) qubits and their colors are complete and not duplicated, necessarily satisfies Eqs.~\eqref{apeq: reduced cond2} and \eqref{apeq: reduced cond3}. 
For example, in Fig.~\ref{fig: subgroups}, we can choose the qubits marked by the red stars as $P=\{1,2,3\}$ with $H=\{e,\sigma_3\}$, reproducing the third layer in Fig.~\ref{fig:pic1}.

After determining $\{ (H_\lambda^{(L)}, P_\lambda^{(L)}) \}_\lambda$ of the final layer, we set $\{ (H_\lambda^{(\ell)}, P_\lambda^{(\ell)}) \}_\lambda$ layer by layer from $\ell=L$ to $\ell=1$.
Then, we are allowed to perform the following operations on $\{ (H_\lambda^{(\ell)}, P_\lambda^{(\ell)}) \}_\lambda$:
\begin{enumerate}
    \item Merge two qubit subsets $P_1$ and $P_2$ if their accompanying subgroups are the same:
    \begin{align}
        (H, P_1), (H, P_2) \to (H, P_1 \sqcup P_2).
    \end{align}
    \item Change a subgroup $H_1$ to a larger one $H_2$ (i.e., $H_1\leq H_2$):
    \begin{align}
        (H_1, P) \to (H_2, P).
    \end{align}
\end{enumerate}
We can perform these operations repeatedly to construct the $\ell$th circuit splitting from the $(\ell+1)$th one.
This construction trivially satisfies Eq.~\eqref{apeq: reduced cond1} and thus Eq.~\eqref{apeq: cond2}.
Furthermore, Eqs.~\eqref{apeq: reduced cond2} and \eqref{apeq: reduced cond3} hold even after the above two operations.
This is because $\bigsqcup_\lambda P_\lambda$ remains unchanged in these operations for Eq.~\eqref{apeq: reduced cond3}, and the following lemma explains Eq.~\eqref{apeq: reduced cond2}:
\begin{lem}
Let $H_1$ and $H_2$ be subgroups of $G$ satisfying $H_1\leq H_2$.
For $q\in\Qb$, if $|G(q)|/|H_1(q)|=s_1$ holds, then $|G(q)|/|H_2(q)|=s_2$ also holds, where we have defined $s_1=|G|/|H_1|$ and $s_2=|G|/|H_2|$.
\end{lem}

\begin{proof}
Given that $H_1$ is a subgroup of $H_2$, any coset of $H_2$ is a disjoint union of cosets of $H_1$ as
\begin{align}
    C^{H_2}_i = \bigsqcup_{j\in \sigma_i} C^{H_1}_j \label{apeq: lem1 eq1}
\end{align}
with $\sigma_i \subseteq [s_1]$ ($|\sigma_i|=|H_2|/|H_1|$).
Meanwhile, by Theorem~\ref{thm: no overlap}, $|G(q)|/|H_1(q)|=s_1$ leads to 
\begin{align}
    G(q) = C^{H_1}_1(q) \sqcup \cdots \sqcup C^{H_1}_{s_1}(q), \label{apeq: lem1 eq2}
\end{align}
where $C^{H_1}_i(q) \cap C^{H_1}_j(q) = \varnothing$ holds for $i\neq j$.
Combining Eqs.~\eqref{apeq: lem1 eq1} and \eqref{apeq: lem1 eq2}, we have
\begin{align}
    G(q) = C^{H_2}_1(q) \sqcup \cdots \sqcup C^{H_2}_{s_2}(q),
\end{align}
where $C^{H_2}_i(q) \cap C^{H_2}_j(q) = \varnothing$ for $i\neq j$.
Since $|C^{H_2}_i(q)|=|H_2(q)|$ for all $i$, we obtain $|G(q)|=s_2 |H_2(q)|$.
\end{proof}

In Fig.~\ref{fig: subgroups}, we can construct earlier layers by modifying the subgroup as $H=\{e,\sigma_3\}\to \{e,(c_4)^2, \sigma_3,\sigma_4\} \to D_4$ with $P=\{1,2,3\}$ being fixed.
This reproduces the second and the first layers in Fig.~\ref{fig:pic1}.

\begin{figure*}[t]
	\centering
	\includegraphics[width=\linewidth]{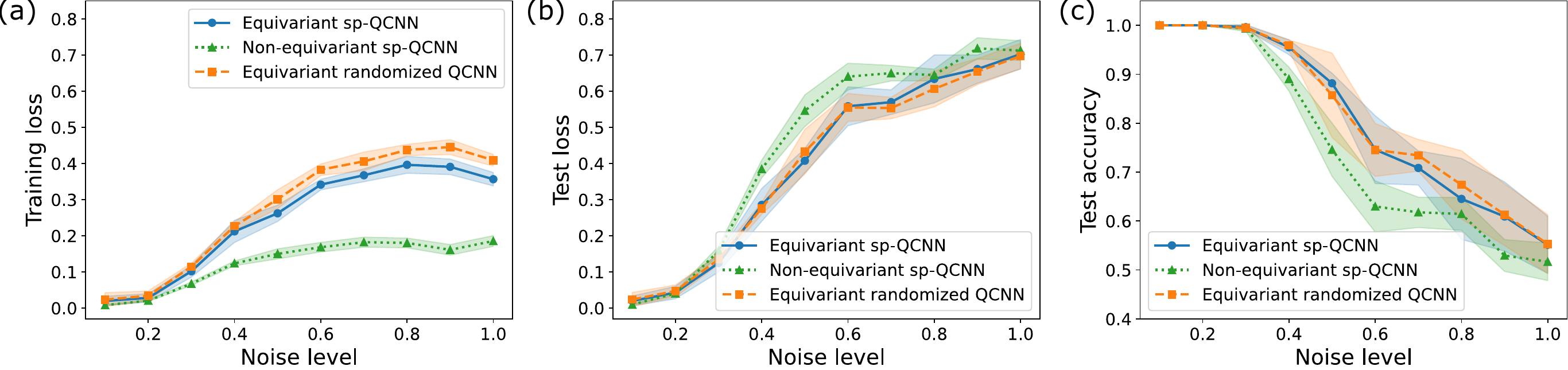}
	\caption{
(a)--(c)  Training loss, test loss, and test accuracy after sufficiently long training processes with $2N_t=20$ for several noise levels $\gamma$ on quantum data.
The circles, triangles, and squares denote the results for equivariant sp-QCNN, non-equivariant sp-QCNN, and equivariant randomized QCNN, respectively. 
The shaded areas indicate the standard deviation for the twenty sets of random initial parameters.  
	}
	\label{fig: noise level}
\end{figure*}

\section{Unitary circuit of numerical experiment} \label{ap: unitary ansatz}

Here, we describe the details of the unitary circuit in the numerical experiment.
For convenience, we assign each qubit a number from one to eight, as shown in Fig.~\ref{fig: demo_lattice}.
The unitary circuit of the equivariant sp-QCNN has three convolutional (or fully-connected) layers as
\begin{align}
    U(\bmt) =  V^{(3)}(\bmt^{(3)})V^{(2)}(\bmt^{(2)})V^{(1)}(\bmt^{(1)}).
\end{align}
Each convolutional layer consists of multiple branches:
\begin{align}
    &Q_1^{(1)} = \{1,2,3,4,5,6,7,8\}, \\
    &Q_1^{(2)} = \{1,4,6,7\}, \\
    &Q_2^{(2)} = \{2,3,5,8\}, \\
    &Q_1^{(3)} = \{1,4\}, \\
    &Q_2^{(3)} = \{6,7\}, \\
    &Q_3^{(3)} = \{2,3\}, \\
    &Q_4^{(3)} = \{5,8\}.
\end{align}
The first convolutional layer $V^{(1)}=V^{(1)}_1$ is given by 
\begin{align}
    &V_1^{(1)}(\bmt) \notag \\
    &= \prod_{i=1}^{d_1} \left( \left( \prod_{\braket{j,k}\in Q_1^{(1)}} R_{j,k}(\delta_{j,k}^i) \right) \left( \prod_{j\in Q_1^{(1)}} R_j(\bma_j^i) \right) \right), 
\end{align}
where we have defined
\begin{align}
    &R_j(\bma)=R_{X_j}(\alpha_1)R_{Z_j}(\alpha_2)R_{X_j}(\alpha_3), \\
    &R_{j,k}(\delta)=R_{Z_jZ_k}(\delta).
\end{align}
The $\braket{j,k}$ denotes the nearest neighbor qubit pair on the $2\times2\times2$ cubic lattice.
To ensure the equivariance, some rotation gates share the parameter values with other rotation gates as
\begin{align}
    &\delta_{1,2}^i = \delta_{2,4}^i = \delta_{4,3}^i = \delta_{3,1}^i, \\
    &\delta_{5,6}^i = \delta_{6,8}^i = \delta_{8,7}^i = \delta_{7,5}^i, \\
    &\delta_{1,5}^i = \delta_{2,6}^i = \delta_{3,7}^i = \delta_{4,8}^i, \\
    &\bma_1^i=\bma_2^i=\bma_3^i=\bma_4^i, \\
    &\bma_5^i=\bma_6^i=\bma_7^i=\bma_8^i.
\end{align}
The second convolutional layer $V^{(2)}=V^{(2)}_2 V^{(2)}_1$ is given by
\begin{align}
    &V_1^{(2)}(\bmt) \notag\\
    &= \prod_{i=1}^{d_2} \left( \left( \prod_{\braket{j,k}\in P_2} R_{j,k}(\delta_{j,k}^i) \right) \left( \prod_{j\in Q_1^{(2)}} R_j(\bma_j^i) \right)\right), \\
    &V_2^{(2)}(\bmt) = U_{c_4} V_1^{(2)}(\bmt) U_{c_4}^\dag,
\end{align}
where we have defined $P_2=\{(1,6),(6,4),(4,7),(7,1)\}$.
The parameters are shared as
\begin{align}
    &\delta_{1,6}^i = \delta_{6,4}^i = \delta_{4,7}^i = \delta_{7,1}^i, \\
    &\bma_1^i=\bma_4^i, \\
    &\bma_6^i=\bma_7^i.
\end{align}
The third convolutional (or fully-connected) layer $V^{(3)}=V^{(3)}_4 V^{(3)}_3 V^{(3)}_2 V^{(3)}_1$ is given by
\begin{align}
    &V_1^{(3)}(\bmt) = \prod_{i=1}^{d_3} R_{1,4}(\delta_{1,4}^i) R_1(\bma_1^i)R_4(\bma_4^i), \\
    &V_2^{(3)}(\bmt) = \prod_{i=1}^{d_3} R_{6,7}(\delta_{6,7}^i) R_6(\bma_6^i)R_7(\bma_7^i), \\
    &V_3^{(3)}(\bmt) = U_{c_4} V_1^{(3)}(\bmt) U_{c_4}^\dag, \\
    &V_4^{(3)}(\bmt) = U_{c_4} V_2^{(3)}(\bmt) U_{c_4}^\dag,
\end{align}
with parameter sharing
\begin{align}
    &\bma_1^i=\bma_4^i, \\
    &\bma_6^i=\bma_7^i.
\end{align}

\vspace{0.3cm}

\section{Numerical analysis of noise resilience} \label{ap: noise resilience}

We investigate the noise effect on the classification performance.
Figure~\ref{fig: noise level} shows numerical results for several noise levels $\gamma$ on quantum data (it is fixed at $\gamma=0.4$ in the main text).
In a wide range of noise levels, we observe the superior generalization performance of equivariant sp-QCNN and randomized QCNN, i.e., lower test loss and higher test accuracy, compared to the non-equivariant sp-QCNN.
Specifically, all the models achieve almost 100\% test accuracy for small $\gamma$, but the high generalization performance of equivariant models becomes obvious as the noise level $\gamma$ increases: the decreases in the test accuracy of the equivariant models due to noise are more modest than that of the non-equivariant model.
These results indicate that the equivariance can improve the model's generalization regardless of noise levels, while we need a more expressive model and more training data for better performance when noise is strong.

\input{source/260416_arxiv_v2.bbl}

\end{document}

%% file: source/260416_arxiv_v2.bbl
%